\documentclass[journal,twoside]{IEEEtran}

\usepackage{hyperref} 
\hypersetup{colorlinks=true, citecolor=blue}
\usepackage{lettrine}
\usepackage{amsfonts}
\usepackage{amsmath,amssymb}
\usepackage{epsfig,graphics,color}
\usepackage[shortlabels]{enumitem}
\usepackage{setspace}
\usepackage{multirow}
\usepackage{wasysym}
\usepackage{cite}
\usepackage{algorithm, algpseudocode}
\usepackage{amsmath,amsthm,amssymb}
\usepackage{psfrag}
\usepackage{cases}
\usepackage{mathrsfs}
\ifCLASSOPTIONcompsoc
\usepackage[tight, normalsize, sf, SF]{subfigure}
\else
\usepackage[tight, footnotesize]{subfigure}
\fi
\usepackage{algorithm}
\usepackage{epstopdf}
\usepackage{footnote}
\usepackage{tikz}
\usepackage{capt-of}
\allowdisplaybreaks
\usetikzlibrary{shapes,arrows}

\graphicspath{{./Figures/}}

\ifCLASSOPTIONcompsoc
\else
\fi

\newtheorem{lemma}{Lemma}

\newtheorem{proposition}{Proposition}

\newtheorem{definition}{Definition}
\newtheorem{remark}{Remark}

\algrenewcommand\algorithmicrequire{\textbf{Input:}}
\algrenewcommand\algorithmicensure{\textbf{Output:}}

\usepackage[normalem]{ulem}

\title{Online Resource Procurement and Allocation in a Hybrid Edge-Cloud Computing System}
\author{Thinh~Quang~Dinh,~\IEEEmembership{Student Member,~IEEE,}  Ben Liang,~\IEEEmembership{Fellow,~IEEE},\\
	Tony~Q.S.~Quek%
	,~\IEEEmembership{Fellow,~IEEE} and Hyundong Shin,~\IEEEmembership{Senior Member,~IEEE}
	\thanks{T. Q. Dinh was with Singapore University of Technology and Design, Singapore 487372. He is now with Trusting Social, Ho Chi Minh City, Vietnam. Emails (e-mail: quangthinh\_dinh@alumni.sutd.edu.sg).}
	\thanks{B. Liang is with the Department of Electrical and Computer Engineering, University of Toronto, Toronto, ON, M5S 3G4, Canada. Emails (e-mail: liang@ece.utoronto.ca.).}	
	\thanks{T. Q. S. Quek is with Singapore University of Technology and Design, Singapore 487372 and also with the Department of Electronics Engineering,
		Kyung Hee University, Yongin-si,   Gyeonggi-do, 17104, Korea (e-mail: tonyquek@sutd.edu.sg).}		
	\thanks{H. Shin is with the Department of Electronic Engineering, Kyung Hee University, Yongin 17104, South Korea (e-mail: hshin@khu.ac.kr).}
}

\begin{document}
	
	\maketitle
	
	\begin{abstract}
		By  {acquiring} cloud-like capacities at the edge of a network, edge computing  is expected to significantly improve user experience. In this paper, we formulate a hybrid edge-cloud computing system where an edge device with limited local resources can rent more from a cloud node and perform resource allocation to serve its users. The resource procurement and allocation decisions depend not only on the cloud's multiple rental options but also on the edge's local processing cost and capacity. We first propose an offline algorithm whose decisions are made with full information of future demand. Then, an online algorithm is proposed where the edge node makes irrevocable decisions in each timeslot without future information of demand. We show that both algorithms have constant performance bounds from the offline optimum. Numerical results acquired with Google cluster-usage traces indicate that the cost of the edge node can be substantially reduced by using the proposed algorithms,  {up to $80\%$ in comparison with baseline algorithms}. We also observe how the cloud's pricing structure and edge's local cost influence the procurement decisions.  
	\end{abstract}
	\begin{IEEEkeywords}
		Mobile edge computing, resource management, competitive analysis
	\end{IEEEkeywords}	
	
	\IEEEpeerreviewmaketitle
	
	\tikzstyle{decision} = [diamond, draw, fill=blue!20,
	text width=4.5em, text badly centered, node distance=3cm, inner sep=0pt]
	\tikzstyle{block} = [rectangle, draw, fill=blue!20,
	text width=5em, text centered, rounded corners, minimum height=4em]
	\tikzstyle{line} = [draw, -latex']
	\tikzstyle{cloud} = [draw, ellipse,fill=red!20, node distance=3cm,
	minimum height=2em]
	\tikzstyle{process} = [rectangle, minimum width=3cm, minimum height=1cm, text centered, draw=black, fill=white]

	\section{Introduction} \label{sec:intro}
	\lettrine{W}ithin the last decade, we have witnessed tremendous growth of data as well as the emergence of the Internet of Things. As a consequence, there is  {an outburst} of digital business that  {utilizes} more and more complex applications with heterogeneous resource requirements. To satisfy the increasing demand of computational power, among  {contemporary} solutions, cloud computing is favored due to its high scalability, accessibility, and availability that come with low storage and computing costs \cite{Luong2017Survey}. Cloud providers offer Infrastructure-as-a-Service (IaaS), which is a form of cloud computing that provides instances of virtualized physical resources, generally termed virtual machines (VMs). For example,  Amazon EC2 \cite{Amazon:price} and Microsoft Azure \cite{Azure:price} are two such services.
	\IEEEpubidadjcol
	There are commonly two pricing options to rent virtual resources: \emph{on-demand} and \emph{reservation} \cite{Amazon:price,Azure:price}. In the first  {option}, the accounting is purely based on the number of instance-hours used, while in the second  {one}, the users {pay a reservation fee in advance, i.e., upfront fee,} in exchange for free or  discounted  {resource usage} over a certain period. The on-demand rental is often considered a costly option, while the same can be said about the upfront fee in the reservation option if the reserved instances are not used sufficiently often. For organizations or users, it is important to achieve cost effective resource procurement and allocation of cloud computing resources. 
	
	Resource procurement and allocation in cloud computing environments have been well-studied \cite{Mao:2011:AMC:2063384.2063449,Malawski2012,Abrishami:2013:DWS:2388122.2388265,WangBroker2015,HuICDCS2015,HongSIGMETRIC2011}.  Many of these works focused on resource allocation with only an on-demand pricing model \cite{Mao:2011:AMC:2063384.2063449,Malawski2012,Abrishami:2013:DWS:2388122.2388265}. Since reserved resources are effective within a period of time, reservation introduces time-correlation in the  {decision} of both resource procurement and resource allocation. Hence, considering both on-demand and reservation pricing options increases the complexity of the problem. However, {leveraging} the discount prices offered by the reservation option can lead to substantial cost savings \cite{HongSIGMETRIC2011,WangBroker2015,HuICDCS2015}. 	 {In}  \cite{Mao:2011:AMC:2063384.2063449,Malawski2012,Abrishami:2013:DWS:2388122.2388265,WangBroker2015,HuICDCS2015,HongSIGMETRIC2011}, application owners/organizations were usually assumed to possess no computing or storage capacities.		
	
	Since the storage and computation cost has dramatically decreased over the last decade, cloud-like capacities have been moving toward the edge network. There are similarities in concept among Edge Computing (EC) \cite{Dinh2018,Hsu2018,Dinh2017,Satyanarayanan2017,BenMobile,MaoTut2017}, fog computing \cite{Chiang2016}, and cloudlets \cite{Satyanarayanan2009}, where services providers or peer helpers, with their own computational and storage power, can implement applications at near-user servers, namely edge devices. However, the edge devices' capacities are limited in comparison with cloud providers. Therefore, it is necessary to investigate hybrid edge-cloud computing systems, specifically, how the edge's capacity and its local processing cost affect the  {previously} mentioned resource allocation problem over cloud computing environments.
	
	There are very few existing works considering hybrid edge-cloud system, or hierarchical fog-cloud networks, where edge devices/lower-tier cloud nodes with their limited resource capacities  {need} to cooperate with high-tier ones.  They usually considered free subscription of  IaaS services  or a simple cost model (e.g., purely on-demand pricing) \cite{LinTPDS2017,GuTETC2017,ChenINFOCOM2017,ChenTMC2018,JiaoTON2017Smoothed,ChampatiIWQoS2015}, which are either impractical or costly. Hence, it remains an open question how edge server parameters, such as computation capacity and processing cost,  as well as the public clouds' pricing options, impact the edge resource allocation decision.

	This paper considers a hybrid edge-cloud computing scenario with an edge node and a public cloud node. The edge node has its own VMs. However, because the arriving VM requests can exceed the edge node's capacity, the edge node also rents remote VMs from the cloud and allocates requests to either rented remote VMs or its own VMs. We propose an optimization framework where the edge node  {performs} resource procurement and allocation in order to minimize its long-term operational cost. This scenario allows us to analyze how the edge node's total cost can be improved by its capacity and the cloud node's rental options. We first propose an offline pseudo-polynomial algorithm whose decisions are made with full information of future demand. We then propose an online algorithm where the edge node makes irrevocable decisions without knowing future information. Moreover, the proposed algorithms' performance guarantees are derived. 
	
	The contributions of this work are summarized as follows:
	\begin{itemize}
		\item  An optimization framework is formulated where the edge node exploits its own resources and the cloud's pricing structure  to minimize its long-term operational cost. In the offline setting with full information of future demand, since finding an optimal solution is intractable, we propose a pseudo-polynomial approximation algorithm, which is shown to achieve a 2-approximation ratio.
		\item  We then propose an online algorithm that does not require any information of future demand. A noticeable feature here is that the proposed online strategy makes irrevocable decisions in each time slot. It achieves a constant competitive ratio of $\max \{6 , \frac{2p}{\lambda}\}$, where $p$ and $\lambda$ are two hyper-parameters related to the pricing structure which will be defined later in the paper. 
		\item Through simulations based on Google cluster-usage traces \cite{clusterdata:Reiss2011}, we observe that the edge node can significantly reduce its operational cost when the edge capacity is considered. We also observe the impact of the cloud's pricing structure and edge's processing cost on the procurement decisions.
	\end{itemize}

	The rest of the paper is organized as follows. In Section
	\ref{section:related}, we present the related work. Section \ref{section:sys_model} describes the
	system model and the problem formulation. In Section \ref{section:offline},
	we propose an offline algorithm to solve this problem which has pseudo-polynomial running time. Section \ref{section:online}  {proposes} an online algorithm for this problem and presents its performance guarantee. Section \ref{section:sim}
	discusses the empirical evaluations based on real-world traces. Conclusions are then given in Section \ref{section:conclusion}.
	
	\section{Related Work} \label{section:related}
	Resource procurement and allocation  {have} been well-studied in many existing works on cloud computing. Some works  {focused} on resource allocation  with a simple procurement model (e.g., applying just on-demand pricing) \cite{Mao:2011:AMC:2063384.2063449,Malawski2012,Abrishami:2013:DWS:2388122.2388265}. Mao and Humphrey \cite{Mao:2011:AMC:2063384.2063449}  proposed heuristic workflow scheduling  strategies  which minimized the execution cost of the workflow. They tried to ensure the jobs' execution deadline as a soft constraint. The Dynamic Provisioning Dynamic Scheduling
	algorithm was proposed by Malawski \emph{et al.} \cite{Malawski2012}, which maximized the number of executed workflows under some quality-of-service constraints.  In \cite{Abrishami:2013:DWS:2388122.2388265} with the same objective as in \cite{Mao:2011:AMC:2063384.2063449}, tasks on a partial critical path were allocated on the same instance by Abrishami \emph{et al.}'s algorithms.
	
	On the other hand, other studies focused on procurement  by dealing with multiple pricing options including reserved instances in order to take advantages of discounted prices \cite{HongSIGMETRIC2011,WangBroker2015,HuICDCS2015}.  Wang \emph{et al.} \cite{WangBroker2015} considered one on-demand and one reserved instance options and proposed online cloud instance acquisition algorithms without full information of future demand, while Hu \emph{et al.} \cite{HuICDCS2015} considered multiple options of different reserved instances in a similar online setting. Hong \emph{et al.} \cite{HongSIGMETRIC2011} first proposed a dynamic programming method to rent purely on-demand	instances to reduce their system's \emph{margin costs}, and then proposed another algorithm utilizing both on-demand and reserved instances to achieve their system's optimal true costs with full information of future demand.

	There are few works considering resource allocation in hybrid clouds, edge-cloud, or fog-cloud networks.   Existing works usually considered a simple cost model such as free cloud access or purely on-demand pricing \cite{ChenINFOCOM2017,JiaoTON2017Smoothed,ChenTMC2018,LinTPDS2017,GuTETC2017,ChampatiIWQoS2015}. Chen \emph{et al.} \cite{ChenTMC2018,ChenINFOCOM2017} proposed  semidefinite-programming based algorithms in order to minimize both energy and latency of workloads in a simple hybrid edge-cloud system. Jiao \emph{et al.} \cite{JiaoTON2017Smoothed} considered a task scheduling problem in multi-tier cloud computing system where the system  {could} jointly optimize its own computational and network resources to reduce the resource allocation cost and resource reconfiguration cost. Furthermore, fog-cloud systems helped to improve the performance of current services  {by reducing} latency and bandwidth consumption in online gaming \cite{LinTPDS2017}, or the operational cost of medical cyber-physical systems \cite{GuTETC2017}. All of these works considered only on-demand pricing, while ignoring the available discounts through reservation can lead to a costly design. In this work, we  leverage the local VMs at the edge node and remote cloud VMs in both on-demand and reservation pricing options.

	Our proposed algorithm is  {an online strategy} where the sequence of decisions is irrevocably made without future knowledge \cite{Borodin:1998:OCC:290169}. In our problem, the edge node needs to decide whether to reserve instances at any time, which can be classified as a variant of the ski rental problem \cite{Karlin:1990:CRA:320176.320216}, a class of rent-or-buy problems. The ski rental problem has been expanded in multiple directions such as the Bahncard problem in transportation \cite{FleischerBahnCard}, TCP acknowledgement problem in networking \cite{TCPAck}, and resource acquisition and resource allocation in cloud computing \cite{WangBroker2015,HuICDCS2015}. In \cite{Karlin:1990:CRA:320176.320216,FleischerBahnCard,TCPAck},   a decision maker only deals with a single level of demand. The problem in our scenario is more complex as cloud computing demands are in multiple levels (e.g., multiple VMs) \cite{WangBroker2015,HuICDCS2015}. Dealing with multiple levels of demand, Wang \emph{et al.} \cite{WangBroker2015} reduced their problem into multiple independent two-option  ski rental problems. Hu \emph{et al.}  \cite{HuICDCS2015} considered multiple reserved instance acquisition as a two-dimensional parking-permit problem. However, in previous works \cite{WangBroker2015,HuICDCS2015}, the number of cloud instances  {which can} be rent in each option is assumed to be infinite. In our work, adding the edge's limited capacity changes the structure of the problem, since here when the edge node's capacity is fully occupied, the excess VM requests will be assigned to the public cloud, in either on-demand VMs or reserved VMs. Hence, our design must account for the impact of the edge's capacity, coupled with the cloud's pricing structure, on the procurement and allocation decisions.
	\section{System Model and Problem Statement} \label{section:sys_model}
	
	In this section, we describe the overall system model, explain how the resource at the cloud and the edge is utilized, and state our optimization problem. The commonly used notation  {throughout} this paper is given in Table \ref{notation}.
	
	\subsection{Computing System Model with Edge and Cloud}
	\begin{table}
		\caption{Notation Used Throughout the Paper.} \label{notation}
		\begin{center}
			\renewcommand{\arraystretch}{1.1}
			{
				\begin{tabular}{p{0.15\linewidth}p{0.7\linewidth}}
					
					\hline
					{\bf Notation} & {\bf Definition} \\ \hline
					$i, t$ & index of time\\
					$l$ & index of demand level\\
					$d_t$ & the aggregated demand of arrival  VM requests at $t$\\
					$r_t$ & the number of remote VMs reserved at $t$ \\
					$n_t$ & the numbers of remote reserved VMs that remain active at $t$ \\
					$\gamma$ & the upfront price for a remote reserved VM\\	
					$\theta$ & the discount cost for using a remote reserved VM  {per time slot}\\
					$p'$ & the cost to rent an on-demand VM per  {time slot}\\								
					$\lambda'$ & the physical cost of running one VM at the edge per  {time slot}\\
					$\tau $ & the reservation period of a remote reserved VM  \\
					$a_t^r$ & the number of requests assigned to remote reserved VMs  {at $t$}\\
					$a_t^o$ & the number of requests assigned to remote on-demand VMs  {at $t$}\\	
					$a_t^w$ & the number of requests assigned to the edge node's VMs  {at $t$}\\	
					$w$ & the number of VMs at the edge\\		
					\hline
				\end{tabular}
			}
			
		\end{center} 
	\end{table}%
	
	We consider a multi-tier computing system with one cloud node and one edge node, as  shown in Fig. \ref{fig:simplified_sys}.  The edge node serves  multiple users, such as mobile users or IoT devices, which have computational jobs to be executed.  The system is time slotted. User requests arrive  at the edge node in every time slot.  Let $d_t$ be the  {total} number of  {VMs requested} by users at time $t$.  {For simplicity, we assume that each user request lasts for one time slot. However, our system model can be extended to accommodate user requests that  {last} for multiple time slots. In that case, $d_t$ accounts for all VMs from new and on-going user requests in time slot $t$. However, our analysis neglects the cost of re-assigning these user requests between VMs.}
	
	The edge node  {has its} own VMs to process the requests.  However, since the capacity of the edge node is limited, the edge node  may need to rent remote VMs from the cloud node to scale up its capacity. There are multiple cloud rental options, each of which has a different cost structure. The edge node  {decides} how  {many} remote VMs it should rent and how to assign the arriving VM  requests to its own VMs or the rented VMs.
	\subsection{Cloud's Resource}
	The cloud service  provider   offers the edge node two options to rent  {its} VMs. The first option is called ``on-demand'' where the edge node can immediately rent a VM  {that} lasts for one time slot with an on-demand price $p'$. In the second option, called ``reserved'', if a VM is reserved at time $t$, it will be effective from $t$ to $t + \tau -1$, where $\tau$ is the reservation period.  {Here, $\tau$ is a given value, not a decision variable.} Let $\gamma$ and $\theta$ denote the  {upfront} price of renting a single remote reserved VM and the  {per-slot} cost of using a reserved VM, respectively. Obviously,  {we should have} $0\leq \theta \leq p'$,  {since otherwise there would be no business case for reserved VMs}. Table \ref{tab:Amazon} shows two examples of on-demand and reservation prices in Amazon EC2.  For ease of exposition, we refer to a VM rented with on-demand price as \emph{remote on-demand VM}, and a VM reserved within $\tau$ time slots  \emph{remote reserved VM}. 
	
	Let $r_t \geq 0$ denote the number of new remote reserved VMs that the edge node decides to rent at time $t$. At time $t$, the number of remote reserved VMs  {that remain} active is
	\begin{figure}
		\centering
		\includegraphics[width=0.6\linewidth]{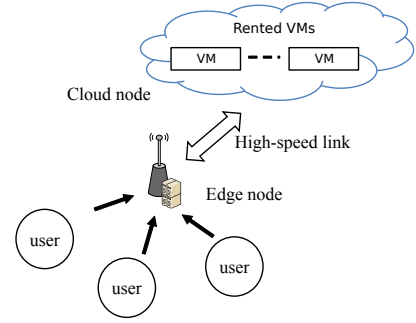} 
		\caption{Computing system with edge and cloud.  \label{fig:simplified_sys}}
	\end{figure}
	\begin{align}
	n_t =   \sum_{i=t - \tau +1}^{t} r_i. \label{eqn:no_effective_VM}
	\end{align}
	Let $a_t^r$ and $a_t^o$ denote the number of VM requests assigned  {to} remote reserved VMs and remote on-demand VMs, respectively.  {Clearly, we have $a_t^r \leq n_t$,   and $a_t^o = 0$ if there are unused reserved VMs, i.e.,  $a_t^r < n_t$.}
	\begin{table} [t]
		\caption{Pricing of on-demand and reserved instances ( Linux, US East ) in Amazon EC2, as of Jan. 10, 2018 \cite{Amazon:price}. }
		\begin{center}
			\footnotesize
			\begin{tabular}{ |c|c|c|c| } 
				\hline
				\bf Instance type & \bf Pricing option & \bf Upfront & \bf Hourly \\
				\hline
				\multirow{2}{4em}{m3.medium} &  On-demand & \$0 & \$0.067 \\ 
				&  1-year reservation & \$211 & \$0.016 \\ 
				\hline
				\multirow{2}{4em}{c3.large} &  On-demand & \$0 & \$0.105 \\ 
				&  1-year reservation & \$326 & \$0.025 \\ 
				\hline			
			\end{tabular} \label{tab:Amazon}	
		\end{center}	
	\end{table}	
	\subsection{Edge's Resource}
	The edge node has its own local capacity $w$, i.e., the number of local VMs  {at} the edge node. We define $a_t^w$ as the number of VM requests the edge node locally processes and $\lambda'$ as the cost of locally processing a unit of VM request.  {The local processing cost is generally defined.}  {As an example, in Section \ref{section:sim}}, we consider  {it} as the electrical cost incurred by physical processors. 	
	\begin{remark} \label{remark:brief_analysis}
		In  {the} extreme case  {where} the edge node's processing cost is greater than or equal to the price of using a remote on-demand VM, i.e., $\lambda' \geq p'$, we should not use the edge node, and should instead allocate VM requests to remote reserved VMs and remote on-demand VMs. Then, the problem is reduced to the  {one} in \cite{WangBroker2015}.
		
		In  the case  where the  usage cost  of  a remote reserved VM is   {greater than or equal to} the edge node's processing cost, i.e., $ {\lambda'} \leq \theta$, a trivial solution is to allocate the VM requests to the edge's VMs first. Then, the excess  {requests} are allocated to either remote reserved VMs or remote on-demand VMs, which  {is again reduced to the problem in} \cite{WangBroker2015}. 
		
		Hence, in this work, we  {only need to} consider the case  {where} $\theta <  {\lambda' < p'}$.
	\end{remark}
	
	%
	\subsection{Problem Formulation} \label{section:sys_model:prob_for}
	
	We consider some time period of system operation $T$, which is assumed to be a multiple of $\tau$, i.e., $T=K\tau$ where $K$ is a positive integer. The user demands over this time period is
	$\{d_1, \dots , d_T \}$. To serve these demands, the edge node decides in each time slot $a^w_t$, $a^o_t$, $r_t$, and $a^r_t$. Then, its total cost is
	\begin{align}
	c = ~ \lambda' \sum_{t=1}^{T}a_t^w + p' \sum_{t=1}^{T} a^o_t + \gamma \sum_{t=1}^{T} r_t + \theta \sum_{t=1}^{T} a^r_t . \label{eqn:cost_func}
	\end{align}
	The first term of (\ref{eqn:cost_func}) is the total cost of processing requests at the edge; the second one is the total  cost of using remote on-demand VMs; the third and the final one are the total costs of reserving  and using remote reserved VMs, respectively.
	\begin{remark} \label{remark:local_proces}
		At each timeslot,  {if $n_t >0$}, the edge node should assign  {new} VM requests to remote reserved  {VMs} first, since $\theta < \lambda' <p'$ as explained in Remark \ref{remark:brief_analysis}. Hence,
		\begin{align}
		a_t^r = \min \{n_t , d_t \}. \label{eqn:allocation_decision:reserved_VM}
		\end{align}	
		{Furthermore, since $\lambda' < p'$, we should always allocate the remaining requests to local processing at the edge before using remote on-demand instances. Hence, we have}
		\begin{align}
		a_t^w = \begin{cases}
		d_t - n_t, \quad &\text{ {if} $0 < d_t -n_t \leq w$},\\
		w, \quad &\text{ {if} $d_t - n_t > w$,}
		\end{cases} \label{eqn:allocation_decision:edge_VM}
		\end{align}
		{and},
		\begin{align}
		a_t^o = (d_t - a_t^w - n_t)^+, \label{eqn:allocation_decision:ondemand_VM}
		\end{align}
		where
		\begin{align}
		{x^+  = \max \{0, x \}}. \nonumber
		\end{align}		
	\end{remark}
	By observing the the relation between $n_t$, $a_t^r$, $a^w_t$  and $a_t^o$ as explained in Remark \ref{remark:local_proces}, we can rewrite (\ref{eqn:cost_func}) as the following:	
	\begin{align}
	 c =&  (\lambda' - \theta) \sum_{t=1}^{T} a_t^w +(p'-\theta) \sum_{t=1}^{T} (d_t - a_t^w - n_t)^+ \nonumber \\
	 &+	\gamma \sum_{t=1}^{T} r_t  + \theta \sum_{t=1}^{T} d_t, \label{eqn:obj_func}
	\end{align} 
	The final term of (\ref{eqn:obj_func}) is {the cost of using only pre-reserved VMs to serve all requests, which is} the minimum cost to process tasks no matter where they are allocated since $\theta < \lambda' < p'$. The first three terms of (\ref{eqn:obj_func}) are the extra costs if VMs are allocated to other VMs. These terms are analogous to the first three terms of (\ref{eqn:cost_func}).
	
	From the above, we see that the edge node only needs to make a sequence of reservation decisions $\mathbf{r}= \{r_1, \dots , r_T \}$ to minimize  {the} total cost, i.e.,
	\begin{align}
	\mathscr P_1:	&\min_{ {\mathbf{r} \in \mathbb{N}^T}} \lambda \sum_{t=1}^{T} a_t^w  + p \sum_{t=1}^{T} (d_t - a_t^w - n_t)^+ + \gamma \sum_{t=1}^{T} r_t  , \label{eqn:new_obj_func}\\
	\mathrm{s.t.} & ~~~(\ref{eqn:no_effective_VM})~ \mathrm{and}~  (\ref{eqn:allocation_decision:edge_VM}), \nonumber
	\end{align}
	where $p:= p' - \theta$ and $\lambda:= \lambda' - \theta$. {Note that} since $\theta \sum_{t=1}^{T} d_t$ is a constant, minimizing (\ref{eqn:new_obj_func}) is equivalent to minimize (\ref{eqn:obj_func}).  
	
	{We note that $\mathscr P_1$ may be viewed as} an extension  {to} the  cloud  instance  acquisition problem in \cite{WangBroker2015},  {where} a cloud broker  rents remote reserved VMs  {and} remote on-demand VMs to serve users' demand. The cloud broker in \cite{WangBroker2015} can be considered as  {an} edge node without local capacity, i.e., $w=0$. In our work, since we consider  a more general edge node with local computing capacity, its capacity and the  local processing cost  affect the edge node's cloud instance procurement decisions.  This substantially alters the structure of the optimization problem and adds to its difficulty. 
	
	 {In $\mathscr P_1$, we focus on the cost at the edge node to serve user demands, without considering the difference in user experience between edge VMs and cloud VMs. However, our formulation is generally applicable. Often the difference in user experience can be negligible, e.g., when the edge node and the cloud are connected by a high-speed link. If it is not negligible, we can modify the cost of edge usage, $\lambda'$, to reflect the priority of VM utilization due to user experience.}  {However, we note also that a decreasing $\lambda'$ pushes the problem $\mathscr P_1$ toward the second case in Remark \ref{remark:brief_analysis}, and when $\lambda' \leq \theta$, the problem is  reduced to the  {one} in \cite{WangBroker2015}.}
	
	Problem $\mathscr P_1$ is combinatorial optimization. It is generally challenging to solve even in the offline setting where the user demands are known in advance. In the more practical online setting, where random user demands arrive dynamically over time,  {it is} even more challenging to design a solution to provide a certain performance guarantee. 	
	
	\subsection{Approximation and Competitive Ratios}
	
	In the following, we state the standard definitions of approximation and competitive ratios, which will be used in the evaluation of the performance of the proposed solution.
	
	\begin{definition}
		Given a sequence of demands $\mathbf{d} = \{d_1, \dots, d_T \}$, let $c^* (\mathbf{d})$ denote the offline optimal cost that could be achieved. {Suppose} an offline algorithm  {achieves} a cost $c^{\mathrm{Off}} (\mathbf{d})$.  {An} approximation ratio $\xi$  {of this offline algorithm} is a constant   {such that for} all possible $\mathbf{d}$,
		\begin{align}
		\nonumber	\frac{c^{\mathrm{Off}} (\mathbf{d})}{c^* (\mathbf{d})} \leq \xi.
		\end{align}
	\end{definition}
	\begin{definition}
		{Given a sequence of demands $\mathbf{d} = \{d_1, \dots, d_T \}$, suppose an online algorithm achieves a cost $c^{\mathrm{On}} (\mathbf{d})$. An competitive ratio $\zeta$ of this online algorithm is  a constant   {such that for} all possible $\mathbf{d}$,}
		\begin{align}
		\nonumber	\frac{c^{\mathrm{On}} (\mathbf{d})}{c^* (\mathbf{d})} \leq \zeta.
		\end{align}	
	\end{definition}	
	Hence,  {the} approximation ratio and competitive ratio are metrics to analyze worst-case performance of offline and online algorithms, respectively. Note that these two ratios are greater than or equal to one.  {Hence}, with $\theta \geq 0$, any ratios obtained with respect to (\ref{eqn:new_obj_func}) still hold with respect to (\ref{eqn:obj_func}). Therefore,  {in this work, we focus on analyzing} the approximation and competitive ratios  {with respect to (\ref{eqn:new_obj_func})}. 
	

	\section{Offline Resource Procurement and Allocation Algorithm} \label{section:offline}
	\begin{figure}
		\centering
		\includegraphics[width=0.8\linewidth]{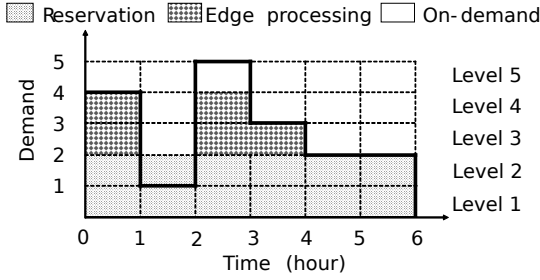} 
		\caption{The resource planning with $\tau=6$ and $T=6$, $r=2$ and $w =2$.  \label{fig:single_interval}}
	\end{figure}

	In this  {section}, we propose an offline  {approximation} algorithm when the demands in all time slots $\mathbf{d} = \{d_1, \dots, d_T\}$  {are} given,  {which has pseudo-polynomial run time}.  The design of  {this} offline algorithm  {will inspire} the online algorithm in Section \ref{section:online}. Furthermore,  {its approximation ratio provides an} intermediate step to derive the competitive  {ratio} of the online algorithm.
	
	\subsection{Algorithm Description}
	
	{We} divide the demands into $ d^{\max}$ levels, where $ d^{\max}$ is the peak demand, i.e., $d^{\max} := \max_t d_t$.	For example, in Fig. \ref{fig:single_interval}, the demands are divided into $5$ levels.  {Let} $d^l_t$  {denote} the demand
	at time $t$ in level $l$, such that
	\begin{align}
	d^l_t =	\begin{cases}
	1 & \text{if $d_t \geq l$},\\
	0 & \text{otherwise.}
	\end{cases}
	\end{align}
	Let utilization $u^l$  denote the number of  {time slots} $d_t$ is greater than or equal to $l$, i.e.,
	\begin{align}
	u^l = \sum_{t=1}^{T} d_t^l. \label{eqn:utilization}
	\end{align}
	We note that since $d_t^l \leq d_t^{l-1}$,	$u^l$ is a non-increasing function with respect to $l$.
	
	The offline algorithm is described as follows. First, we consider the $K$ non-overlapping intervals,  {each of $\tau$ duration}, that comprise the $T$ time period as described in Section \ref{section:sys_model:prob_for}. Let $I_k, k \in \{1,\dots, K\},$ denote  {the intervals}.  The proposed offline algorithm  {decides} how many remote VMs should be  reserved at the beginning of each interval $I_k$,  {i.e., when} $t = (k-1)\tau +1$. Let
	\begin{align}
	u^l_k = \sum_{t \in I_k} d_t^l
	\end{align}
	denote the utilization of level $l$ in $I_k$. 
	
	Consider level $l = 1$. Based on Remark \ref{remark:local_proces}, if a VM is reserved for the demand at this level,  {user} requests from $l=2$ to $w+1$ are allocated to  {edge} VMs and  {the} ones from $l=w+2$ to $d^{\max}$ are allocated to on-demand VMs. Otherwise,  {the} requests from $l=1$ to $w$ are allocated to edge VMs and  {the} ones from $w+1$ to $d^{\max}$ are allocated to on-demand VM. Therefore,   VM reservation is justified if
	\begin{align}
	\nonumber	\gamma + \lambda \sum_{j = 2}^{1+w} u^j_k + p \sum_{j=w+2}^{ d^{\max}}u^j_k \leq \lambda \sum_{j = l}^{w} u^j_k + p \sum_{j = w+1}^{ d^{\max}} u^j_k, 
	\end{align}
	which implies, 
	\begin{align}
	\nonumber	\gamma \leq \lambda u^1_k + (p - \lambda) u_k^{w+1} .
	\end{align}
	More generally, consider level $l$ when $l-1$ VMs are already reserved, the reservation at level $l$ is justified if
	\begin{align}
	\nonumber		 & l \gamma + \lambda \!\!\! \sum_{j = l+1}^{l+w} \!\!\!\! u^j_k + p \!\!\!\!\!\! \sum_{j=l+w+1}^{ d^{\max}}\!\! \!\!\!\! u^j_k \leq (l-1) \gamma
	+ \lambda \!\!\! \sum_{j = l}^{l+w-1} \!\!\!\! u^j_k + p \!\!\!\! \sum_{j = l + w}^{ d^{\max}} \!\!\!\! u^j_k ,
	\end{align} 
	which implies
	\begin{align}
	\nonumber	  \gamma \leq \lambda u^l_k + (p-\lambda)u^{l+w}_k.
	\end{align}
	Therefore, from demand level $l=1$ to $d^{\max}$, the edge node  {reserves one additional} VM at  {each} level $l$  {if and only if}
	\begin{align}
	\gamma \leq \lambda u^l_k + (p-\lambda)u^{l+w}_k. \label{eqn:algorithm1:vital_ineq}
	\end{align}
	
	This algorithm gives the total number of VMs that should be reserved  {for $I_k$. Then}, the  {user} requests are allocated to the three types of VMs according to  {(\ref{eqn:allocation_decision:reserved_VM}), (\ref{eqn:allocation_decision:edge_VM}), and (\ref{eqn:allocation_decision:ondemand_VM})}.	We term this the  Offline Resource Procurement  {and Allocation} Algorithm (OfflineRPAA) and summarize it in Algorithm \ref{algorithm:heu_per_plan}.  {This algorithm has $O(d^{\max}T)$ time complexity and $O(T)$ space complexity, where $T$ is the length of time horizon and $d^{\max}$ is the peak computing demand.}
	
		\begin{algorithm}    [t]                
			\caption{Offline Resource Procurement  {and Allocation} Algorithm (OfflineRPAA)}          
			\label{algorithm:heu_per_plan}                           
			{
				\begin{algorithmic}[1]                    
					\Require Segment $T$ into intervals $\{I_k\}_{k = 1,2, \ldots, K}$, each with length $\tau$. Initiate $r_t := 0$  {for all $t$}. For each segment $I_k$, the edge node knows $d_t, t \in I_k$, the pricing structure's hyper-parameters $\gamma$, $p$, and $\lambda$.
					\Ensure For each $I_k$, we compute reservation decision $r_{(k-1)\tau+1}$, allocation decisions $a^r_t$, $a^o_t$ and $a^w_t$, $t\in I_k$.
					\For  {all  {intervals} $\{I_k\}$}
					\For {$l = 1$ to $l= d^{\max}$}
					\If  {$\gamma \leq \lambda u^l_k + (p-\lambda)u^{l+w}_k$}
					\State {$r_{(k-1)\tau+1} \leftarrow r_{(k-1)\tau+1} +1.$}	
					\EndIf
					\EndFor
					\EndFor
					\State {At each time slot $t$,
						\begin{align}
						\nonumber	& a_t^r = \min \{n_t , d_t \}\\
						\nonumber 	&	a_t^w = \begin{cases}
						(d_t - n_t)^+, \quad &\text{ {if} $0 \leq d_t -n_t \leq w$}\\
						w, \quad &\text{ {if} $d_t - n_t > w$}
						\end{cases}\\
						\nonumber & a_t^o = (d_t - a_t^w - n_t)^+.
						\end{align}
					}
				\end{algorithmic}}
			\end{algorithm}
	
	\subsection{Performance Guarantee}
	In this section, we show  {that} Algorithm \ref{algorithm:heu_per_plan} achieves  {a} $2$-approximation ratio. First,  let $\mathcal{X}' \subset \mathbb{N}^T$ denote the set of solutions in which reservation decisions are made only at the beginning of each interval, i.e.,
	\begin{align}
	\mathcal{X}' = \Big \{  \{ r_t \} \in \mathbb{N}^T |  r_t = 0, & \text{when } t \neq (k-1)\tau +1 , \nonumber \\
	& k \in \{1,\dots, K \}  & \Big \}  .
	\end{align}
	
	We  {will show} that  the solution generated by Algorithm \ref{algorithm:heu_per_plan} achieves the smallest cost among all $\mathbf{r} \in \mathcal{X}'$ for $\mathscr P_1$. Then, we will  {show} that there exists  a solution $\mathbf{r}^f \in \mathcal{X}'$ that achieves  {a} $2$-approximation ratio. Hence, Algorithm \ref{algorithm:heu_per_plan} achieves $2$-approximation ratio.
	
	Finding a  solution $\mathbf{r}  \in \mathcal{X}'$ to minimize $\mathscr P_1$  {is equivalent to the following:}
	\begin{align}
	\nonumber	\mathscr P_2:	\min_{ \mathbf{r}  \in \mathcal{X}'}\sum_{k  = 1}^{K} \Big [ & \lambda \sum_{t  {\in} I_k} a_t^w + p \sum_{t  {\in} I_k} (d_t - a_t^w - n_t)^+  + \gamma \sum_{t  {\in} I_k} r_t  \Big ].
	\end{align} 
	{Since the remote reserved VMs under Algorithm \ref{algorithm:heu_per_plan} span only one interval}, $\mathscr P_2$ can be decomposed into $K$ independent  {sub-problems as follows}.
	\begin{align}
		\mathscr P_2^k:	\min_{ {r_{(k-1)\tau +1}   }} & \gamma r_{(k-1)\tau +1} + p \sum_{t \in I_k} (d_t - a_t^w - n_t)^+ \nonumber \\
		&+ \lambda \sum_{t \in I_k} a_t^w . \label{eqn:objective_value:I_k}
	\end{align}	
	\begin{lemma}\label{lemma:alg_1_optimal}
		Algorithm \ref{algorithm:heu_per_plan} achieves an optimal decision in each $\mathscr P_2^k$. In other words, Algorithm \ref{algorithm:heu_per_plan} provides the lowest cost for $\mathscr P_1$,  {restricted to $\mathbf{r} \in \mathcal{X}'$}.
	\end{lemma}	
	\begin{proof}
		Firstly,  the cost incurred by an algorithm is equal to the sum of the cost of each level of demand. Hence,  an algorithm that provides the lowest sum of all levels' cost is optimal.
		
		Consider an interval $I_k$, let $r^+_k$ denote the optimal number of reserved VMs for $I_k$. From (\ref{eqn:no_effective_VM}), we have $n_t^+ = r^+_k, \forall t \in I_k$. According to Remark \ref{remark:local_proces}, at $t$ with $n_t >0$, we should allocate requests to reserved VMs first. As a result, with $n_t^+ = r^+_k$, we should allocate $d_t$ to $r_k^+$ reserved VMs first for all $t \in I_k$, i.e., all utilizations $u^l_k$ from $l=1$ to  $r^+_k$ are allocated to reserved VMs without any gaps in between.
		
		Consider level $r^+_k+1$ of demand, reserving one more VM at this level increases the cost within $I_k$. Formally, this is expressed by the following inequality.
		\begin{align}
		\nonumber		 (r^+_k \!\! +1 )\gamma + \lambda \!\!\!\!\! \sum_{j = r^+_k +2}^{r^+_k +w+1} \!\!\!\!\! u^j_k + p \!\!\!\!\!\!\!\! \sum_{j=r^+_k+w+2}^{ d^{\max}}\!\! \!\!\!\!\! \!\!u^j_k >& r^+_k \gamma + \lambda \!\!\!\! \sum_{j = r^+_k+1}^{r^+_k+w} \!\!\!\! u^j_k + p \!\!\!\! \sum_{j = r^+_k + w}^{ d^{\max}} \!\! \!\!\!\! u^j_k ,
		\end{align} 
		which implies
		\begin{align}
		\nonumber	  \gamma > \lambda u^{r^+_k+1}_k + (p-\lambda)u^{r^+_k+w+1}_k.
		\end{align}
		Since $u^l$ is a non-increasing function respect to $l$, 
		\begin{align}
		\gamma > \lambda u^{l}_k + (p-\lambda)u^{l+w}_k, \quad \forall l > r^+_k. \label{eqn:algorithm_1:opt:no_more_than_l}
		\end{align}	
		Hence, from (\ref{eqn:algorithm_1:opt:no_more_than_l}) and (\ref{eqn:algorithm1:vital_ineq}), we see that Algorithm \ref{algorithm:heu_per_plan} does not reserve at any levels higher than level $r^+_k$.
		
		{Suppose Algorithm \ref{algorithm:heu_per_plan} reserves VMs from levels $l=1$ to $l'$, and $l' < r^+_k$.} Consider level $l'+i, i \in \mathbb{N}$ such that $l'  < l'+i \leq r^+_k $. According to Algorithm \ref{algorithm:heu_per_plan},  {since a VM is not} reserved at level $ {l'}+i$,   {we have}
		\begin{align}
		\nonumber	\gamma > \lambda u^{l'+i}_k + (p-\lambda)u^{l'+i+w}_k.
		\end{align} 
		However, the optimal solution reserves at level $l$, which implies
		\begin{align}
		\nonumber	\gamma \leq \lambda u^{r^+_k}_k + (p-\lambda)u^{r^+_k+w}_k.
		\end{align} 
		{Thus,}
		\begin{align}
		\lambda u^{l'+i}_k + (p-\lambda)u^{l'+i+w}_k < \lambda u^{r^+_k}_k + (p-\lambda)u^{r^+_k+w}_k. \label{eqn:inequality1}
		\end{align}
		However, since $u^l_k$ is a non-increasing function with respect to $l$, (\ref{eqn:inequality1}) does not hold. As a result, Algorithm \ref{algorithm:heu_per_plan} reserves until reaching level $r^+_k$. The lemma  {is proven}.
	\end{proof}

	Let $c^{\text{Alg 1}}$ denote the cost achieved by Algorithm 1. By applying Lemma \ref{lemma:alg_1_optimal}  {above}, we obtain the following main proposition:
	
		\begin{proposition} \label{proposition:interval-optimal}
			Algorithm \ref{algorithm:heu_per_plan} has 2-approximation ratio, i.e., $c^{\text{Alg 1}} \leq 2c^*$.
		\end{proposition}
		\begin{proof}
			Let $\mathrm{OPT}$ denote the optimal solution where $\mathbf{r}^* = \{r_1^*, \ldots, r_T^* \}$ is the optimal reservation of $\mathscr P_1$. There exists  {a solution $\mathbf{r} ^f \in \mathcal{X}'$, whose reservation decisions $r^f_t$} at $t = (k-1)\tau +1$ are as follows: 
				\begin{align}
				r^f_{(k-1)\tau +1}  =\begin{cases}
				\sum_{i = 1}^{\tau} r^*_i, &\text{if}~~ k= 1,\\
				\sum_{i = (k-2)\tau +1}^{k\tau} r^*_i, &\text{if}~~ k=2,\ldots,K.
				\end{cases}
				\end{align}
			 {We note that $r^f_{(k-1)\tau+1}$ is the sum of the
			optimal reservations in the previous interval $[(k-2)\tau +1, (k-1)\tau ]$ and the current interval $[(k-1)\tau +1, k\tau ]$ when $k >1$.} Then, we have
			\begin{align}
			\sum_{t=1}^{T} r^f_t \leq 2 \sum_{t=1}^{T} r^*_t. \label{eqn:algorithm1:prop1:reserved_VM}
			\end{align}
			
			Let $n_t^*$ and $n^{f}_t$  denote the numbers of remote reserved VMs
				that remain effective at time $t$ of the optimal strategy and $\mathbf{r} ^f$, respectively. Now, we compare $n^{f}_t$ and $n_t^*$. For $t \in \{1, \ldots , \tau \}$, we have
				\begin{align}
				\nonumber n^*_t &= \sum_{i=1}^{t} r^*_i,\\
				\nonumber n^f_t &= \sum_{i=1}^{t} r^f_i \stackrel{\{a\}}{=} r^f_1 = \sum_{i=1}^{\tau} r^*_i,
				\end{align}
				where  {$\stackrel{\{a\}}{=}$} above is because $r_i^f = 0$ if $i \neq 1$. Hence, $n^f_t \geq n^*_t$, when $t \in \{1, \ldots , \tau \}$. For $t = (k-1)\tau +j$, where $j \in \{1,\ldots, \tau \}$ and $k >1$, we have
				\begin{align}
				\nonumber n^*_{(k-1)\tau + j} & = \sum_{i=(k-2)\tau + j  +1}^{(k-1)\tau+j} r^*_i,\\
				\nonumber n^f_{(k-1)\tau + j} & = \sum_{i=(k-2)\tau + j +1}^{(k-1)\tau +j} r^f_i \stackrel{\{b\}}{=} r^f_{(k-1)\tau + 1} = \sum_{i=(k-2)\tau + 1}^{k\tau} r^*_i,
				\end{align}			
				where  {$\stackrel{\{b\}}{=}$} above is because $r_t^f = 0$ if $t \neq (k-1)\tau +1$. For any $j \in \{1,\ldots, \tau \}$, we have $(k-2)\tau + 1 \leq (k-2)\tau + j  +1$ and $k\tau \geq (k-1)\tau +j$. Hence we have
				\begin{align}
				\nonumber	n^{f}_t \geq n_t^*,~~\text{for}~~t = (k-1)\tau +j, j \in \{1,\ldots, \tau \}.
				\end{align}	Hence, 
			\begin{align}
			n^{f}_t \geq n_t^*, \quad \forall t. \label{eqn:algorithm1:prop1:effective_VM}
			\end{align}
			From (\ref{eqn:algorithm1:prop1:effective_VM}), according to Remark \ref{remark:local_proces}, for both OPT and Algorithm \ref{algorithm:heu_per_plan}, given the same demand $\mathbf{d}$, they both allocate requests to remote reserved VMs first, local VMs second, and then on-demand VM last. Therefore, we achieve
			\begin{align}
			a_t^{w^{f}} \leq a_t^{w^*}, \label{eqn:algorithm1:prop1:local_assign}
			\end{align}
			and
			\begin{align}
			a^{o^f}_t = (d_t - a_t^{w^{f}} - n_t^{f})^+ \leq (d_t - a_t^{w^*} - n_t^*)^+ = a^{o^*}_t, \forall t. \label{eqn:algorithm1:prop1:ondemand_assign}
			\end{align}			
			
			Let $c^*$ and $c^{f}$ be the   objective values of (\ref{eqn:new_obj_func}) of the optimal solution and $\mathbf{r} ^f$, respectively. From (\ref{eqn:algorithm1:prop1:reserved_VM}), (\ref{eqn:algorithm1:prop1:local_assign}), and (\ref{eqn:algorithm1:prop1:ondemand_assign}), we have
			\begin{align}
			\nonumber	 c^{f} &= \gamma \sum_{t=1}^{T} r_t^{f} + p \sum_{t=1}^{T} (d_t - a_t^{w^f} - n_t^{w^f})^+ + \lambda \sum_{t=1}^{T} a_t^{w^f} \\
			\nonumber	 & \leq 2\gamma \sum_{t=1}^{T} r_t^* + p \sum_{t=1}^{T} (d_t - a_t^{w^*} - n_t^*)^+ + \lambda \sum_{t=1}^{T} a_t^{w^*} \\
			\nonumber	 & \leq 2 c^*.
			\end{align}
			According to Lemma \ref{lemma:alg_1_optimal}, since $\mathbf{r}^f \in \mathcal{X}'$, we have $c^{\text{Alg 1}} \leq c^f$. Hence, $c^{\text{Alg 1}} \leq c^f \leq 2c^*$. This proves the proposition.
		\end{proof}
		{We note that although the approximation ratio that Algorithm \ref{algorithm:heu_per_plan} achieves is similar to that of the  offline algorithm proposed in \cite{WangBroker2015}, the proof of Proposition \ref{algorithm:heu_per_plan} is substantially different and utilizes Remark \ref{remark:local_proces} with the assumption of $\theta \leq \lambda' \leq p'$}. 
		\section{Online Resource Procurement  {and Allocation} Algorithm} \label{section:online}
		
		In this section, we consider an online strategy without any prior knowledge about the future demand. We keep track  {of} the past demand of users and make decision at each time slot $t$  {after the current demand arrives}.
		\subsection{Algorithm Description}
		Inspired by Algorithm \ref{algorithm:heu_per_plan} and Proposition \ref{proposition:interval-optimal}, we again divide the  {time axis} into  intervals of length $\tau$ timeslots. Within any of such interval $I_k$, at each time slot $t\in I_k$, i.e.,  $t = (k-1) \tau + i$ for $i \in  \{1,\dots,\tau \}$, we dynamically update the sequence of reservation decisions $\{r_t\}$ from the current time slot to the end of the interval. Note that since  the edge node makes irrevocable reservations, at each time slot $t \in I_k$, we can only update $\{r_{t'} \}$ for $t' \in \{t, \ldots, k\tau \}$. In addition, the value of $r_{t}$ can only increase or remain unchanged, as we make new reservation decisions in each timeslot. 
		
		Our decision is based on the history of demand $d_t$, the number of previously added reserved VMs $r_t$, and the number of remaining active reserved VMs $n_{t'}$ for $t' \in \{t,\ldots, t+\tau -1 \}$. Similarly to how the offline Algorithm \ref{algorithm:heu_per_plan} uses (\ref{eqn:algorithm1:vital_ineq}), in the online algorithm, the edge node will  reserve a VM at level $l$ if
		\begin{align}
		\gamma \leq \lambda \sum_{i = (k-1)\tau + 1}^{t} d_i^l + (p - \lambda) \sum_{i = (k-1)\tau + 1}^{t} d_i^{l+w}. \label{eqn:online_thres}
		\end{align}
		Then, the VM requests are allocated to the three types of VMs according to Remark \ref{remark:local_proces}. Note that (\ref{eqn:algorithm1:vital_ineq}) suggests the edge node should reserve at level $l$ if the gain of reserving one more VM is higher than its upfront cost. In contrast, (\ref{eqn:online_thres}) suggests the edge node should reserve if the gain of the hypothetical scenario, where the edge node had reserved a VM at the beginning of interval $I_k$, is higher than the upfront cost. 
		
		We further note that in some intervals, if  {a remote reserved VM is procured} at level $l$ while there  {is} no reserved instance at level $l-1$, the reserved  {instance} is assigned  {at} $l-1$. Moreover, while there is already a VM reserved at level $l$, the edge node will postpone to reserve at level $l$ until the VM expires. The resultant algorithm is termed the Online Resource Procurement and Allocation Algorithm (Online RPPA) and is summarized in Algorithm \ref{algorithm:online-deter}, which is run continuously at each timeslot $t$. 
		
		We note in particular that, for each time $t$, with the knowledge of the number of remaining active reserved VM $n_{t'}$ from $t' = 1$ to $t+\tau -1$, if (\ref{eqn:online_thres}) is satisfied, then, from lines $4$ to $8$ of Algorithm \ref{algorithm:online-deter}, we inspect the number of reservations from the current time $t$ to the end of $I_k$, and then we reserve a VM at level $l$ only if there exists a $t'$ such that $n_{t'} < l$. Under this procedure, we reserve at most one VM within $I_k$ at any given level $l$, to avoid  redundant reservations.
		
		 {Since $T$ goes to infinity, we consider the complexity of Algorithm \ref{algorithm:online-deter} in a single time slot. Assuming that $d^{\max}$ is given, it has $O(d^{\max}\tau)$ time complexity and $O(\tau)$ space complexity.}
		
		\begin{algorithm}    [t]                
			\caption{Online Resource Procurement  {and Allocation} Algorithm (OnlineRPAA)}          
			\label{algorithm:online-deter}                           
			{
				\begin{algorithmic}[1]                    
					\Require demand $d_t$, corresponding $I_k$, previously added reserved VM $r_{t'}$ for $t' \in \{t, \ldots, k\tau \}$ and the number of remaining active reserved VM $n_{t'}$ for $t' \in \{t,\ldots, t+\tau-1 \}$, the pricing structure's hyper-parameters $\gamma$, $p$, $\lambda$
					\Ensure updated reservation decisions $r_t'$ with $t' \in \{t,\ldots, k\tau\}$, allocation decisions $a^r_t$, $a^o_t$ and $a^w_t$
					\For  {$l=1$ to $d_t$}
					\If {$\gamma \leq \lambda \sum_{i = (k-1)\tau + 1}^{t} d_i^l + (p - \lambda) \sum_{i = (k-1)\tau + 1}^{t} d_i^{l+w}$}
					\For {$t' = t$ to $k \tau$}
					\If {$n_{t'} < l$}
					\State $r_{t'} \leftarrow r_{t'} +1.$
					\EndIf
					\EndFor
					\EndIf
					\EndFor
					\State {The VM requests are assigned as follows:
						\begin{align}
						\nonumber	& a_t^r = \min \{n_t , d_t \}\\
						\nonumber 	&	a_t^w = \begin{cases}
						(d_t - n_t)^+, \quad &\text{ {if} $0 \leq d_t -n_t \leq w$}\\
						w, \quad &\text{ {if} $d_t - n_t > w$}
						\end{cases}\\
						\nonumber & a_t^o = (d_t - a_t^w - n_t)^+.
						\end{align}
					}
				\end{algorithmic}}
			\end{algorithm}
			\subsection{Performance Guarantee}
			
			Among the intervals of $\tau$ timeslots defined above, let $\mathcal{I}_{\text{cheap}}$ denote the set of intervals  in which Algorithm \ref{algorithm:heu_per_plan} does not reserve any remote VMs, and  let $\mathcal{I}_{\text{expensive}}$ denote the set of intervals  in which Algorithm \ref{algorithm:heu_per_plan} reserves at least one remote VM. Recall that the objective value of (\ref{eqn:new_obj_func}) achieved by Algorithm \ref{algorithm:heu_per_plan} is denoted by $c^{\text{Alg 1}}$. 	We further let  $c^{\text{Alg 2}}$ denote the part of the objective value of (\ref{eqn:new_obj_func}) resulting from Algorithm \ref{algorithm:online-deter}.	We also  {let $c_{k}$ denote} the objective values of (\ref{eqn:new_obj_func}) within $I_k$
			\begin{align}
			c_{k} =  \gamma \sum_{t \in I_k} r_t + \lambda \sum_{t \in I_k} a_t^w  + p \sum_{t \in I_k} (d_t - a_t^w - n_t)^+ .  \label{eqn:new_obj_func:I_k}
			\end{align}		
			Since $I_k$ are non-overlapping intervals, $c^{\text{Alg 1}} = \sum_{k} c^{\text{Alg 1}}_k$ and $c^{\text{Alg 2}} = \sum_k c^{\text{Alg 2}}_k$.
			
			\begin{lemma}
				Within $I_k \in \mathcal{I}_{\text{cheap}}$, $c^{\text{Alg 2}}_k \leq c^{\text{Alg 1}}_k$. \label{lemma:cheap_phase}
			\end{lemma} 
			
			\begin{proof}
				Algorithm \ref{algorithm:heu_per_plan} does not reserve any VMs within any $I_k \in \mathcal{I}_{\text{cheap}}$, i.e., $r_t^{\text{Alg 1}} = n_t^{\text{Alg 1}} = 0, \forall t \in I_k$. From  (\ref{eqn:online_thres}), Algorithm \ref{algorithm:online-deter} also does not reserve any VM within $I_k$, i.e., $r_t^{\text{Alg 2}} = 0, \forall t \in I_k$, which is equivalent to
				\begin{align}
				r_t^{\text{Alg 1}} = r_t^{\text{Alg 2}} = 0, \forall t \in I_k. \label{eqn:lemma:I_cheap:reserve_decision}
				\end{align}
				However, Algorithm \ref{algorithm:online-deter} can reserve at any time $t$, not just at the beginning of each $I_k$. Therefore, there may be some remaining active reserved VMs from the previous interval (which necessarily is in $\mathcal{I}_{\text{expensive}}$), i.e., $n_t^{\text{Alg 2}} \geq n_t^{\text{Alg 1}}$. Hence, similar to how we obtain (\ref{eqn:algorithm1:prop1:local_assign}) and (\ref{eqn:algorithm1:prop1:ondemand_assign})  from (\ref{eqn:algorithm1:prop1:effective_VM}), we have, $\forall t \in I_k$,
				\begin{align}
				\nonumber 	a_t^{w^{\text{Alg 2}}} &\leq a_t^{w^{\text{Alg 1}}},\\
				\nonumber 	a^{o^{\text{Alg 2}}}_t  & = (d_t - a_t^{w^{{\text{Alg 2}}}} - n_t^{{\text{Alg 2}}})^+ \\
				& \leq (d_t - a_t^{w^{\text{Alg 1}}} - n_t^{\text{Alg 1}})^+ = a^{o^{\text{Alg 1}}}_t. \label{eqn:lemma:I_cheap:local_ondemand_allocation}
				\end{align}
				Moreover, by substituting (\ref{eqn:lemma:I_cheap:reserve_decision}) and (\ref{eqn:lemma:I_cheap:local_ondemand_allocation}) into (\ref{eqn:objective_value:I_k}), we have 
				\begin{align}
				\nonumber &	p \sum_{t \in I_k} (d_t - a_t^{w^{\text{Alg 2}}} - n_t^{\text{Alg 2}})^+ + \lambda \sum_{t \in I_k}  a_t^{w^{\text{Alg 2}}}  \\  
				\nonumber		  \leq &  p \sum_{t \in I_k} (d_t - a_t^{w^{\text{Alg 1}}} - n_t^{\text{Alg 1}})^+ + \lambda \sum_{t \in I_k} a_t^{w^{\text{Alg 1}}},
				\end{align}
				which is equivalent to
				\begin{align}
				c^{\text{Alg 2}}_k \leq c^{\text{Alg 1}}_k.
				\end{align}
			\end{proof}

			Next, we consider an arbitrary $I_k \in \mathcal{I}_{\text{expensive}}$. For 
			this case, we need to define the cost of a reservation strategy (i.e., 
			the objective values of (\ref{eqn:new_obj_func}) for either Algorithm \ref{algorithm:heu_per_plan} or Algorithm \ref{algorithm:online-deter}) at 
			level $l$ within interval $I_k$ as follows:
			
				\begin{align}
				\nonumber	c_{k,l} = & \sum_{t \in I_k} \Big[ \gamma \mathbb{I} ( r_t \geq l ) + d_t^l \Big( \lambda \mathbb{I}( n_t < l \leq n_t +w )   \nonumber \\
				& +   p \mathbb{I} (l > n_t +w ) \Big)  \Big]\\
				= &\gamma  \sum_{t \in I_k}  \mathbb{I} ( r_t \geq l )  + \lambda \sum_{t \in I_k}d_t^l  \mathbb{I}( n_t < l \leq n_t +w ) \nonumber \\ 
				& {+} p \sum_{t \in I_k} d_t^l  \mathbb{I} (l > n_t +w ),
				\label{eqn:cost:each_lvl}
				\end{align}			
				where $\mathbb{I} (\cdot)$ is the indicator function.	
			
			The following lemma indicates that this definition of $c_{k,l}$ properly 
			separates the total cost by demand levels:
 		
			\begin{lemma} \label{lemma:objective_function:seperated_cost}
				In any interval $I_k$, $c_k = \sum_l c_{k,l}$.
			\end{lemma}
		
		\begin{proof}
				Consider the first element of (\ref{eqn:cost:each_lvl}) and (\ref{eqn:new_obj_func:I_k}), we have
				\begin{align}
				\nonumber	& \sum_l	\gamma  \sum_{t \in I_k}  \mathbb{I} ( r_t \geq l )\\
				\nonumber  =& \gamma \sum_l \Big (\mathbb{I} (r_{(k-1)\tau +1}\geq l ) + \ldots +  \mathbb{I} (r_{k\tau}\geq l)    \Big)\\
				\nonumber	=& \gamma \Big ( \sum_l \mathbb{I} (r_{(k-1)\tau +1}\geq l ) + \ldots    + \sum_l \mathbb{I} (r_{k\tau}\geq l)  \Big) \\
				\nonumber	=& \gamma \Big ( r_{(k-1)\tau +1}\ + \ldots + r_{k\tau}  \Big) \\
				=&  \gamma \sum_{t \in I_k} r_t. \label{eqn:obj_func:I_k:reserved}
				\end{align}
				Consider the second element of (\ref{eqn:cost:each_lvl}) and (\ref{eqn:new_obj_func:I_k}), we have	
				\begin{align}
				\nonumber	&\sum_l	\lambda \sum_{t \in I_k}d_t^l  \mathbb{I}( n_t < l \leq n_t +w )\\
				\nonumber  =& \lambda \sum_{t \in I_k} \Big [ (d_t - n_t)  \mathbb{I} ( n_t < d_t \leq n_t +w )  + w  \mathbb{I} ( d_t > n_t +w ) \Big]\\
				= & \lambda \sum_{t \in I_k} a^w_t. \label{eqn:obj_func:I_k:loca}
				\end{align}
				Consider the last element of (\ref{eqn:cost:each_lvl}) and (\ref{eqn:new_obj_func:I_k}), we have
				\begin{align}
				 &\sum_l	p \sum_{t \in I_k} d_t^l  \mathbb{I} (l > n_t +w ) \nonumber  \\
				 =&  p \sum_{t \in I_k} (d_t - w - n_t) 
				\mathbb{I} ( d_t > n_t +w ) . \label{eqn:obj_func:I_k:ondemand_a}
				\end{align}	
				Moreover, when $n_t < l \leq n_t +w $, we have
				\begin{align}
				\nonumber		a^w_t &= d_t -n_t .
				\end{align}
				Hence, 
				\begin{align}
				\nonumber	& p \sum_{t \in I_k} (d_t - w - n_t) \mathbb{I} ( d_t > n_t +w )\\
				\nonumber	 =&  p \sum_{t \in I_k} (d_t - w - n_t) \mathbb{I} ( d_t > n_t +w )\\
				\nonumber &+ p \sum_{t \in I_k} (d_t - a_t^w - n_t)  \mathbb{I}( n_t < d_t \leq n_t +w )\\
				= & p \sum_{t \in I_k} (d_t - a_t^w - n_t) ^+.\label{eqn:obj_func:I_k:ondemand_b}
				\end{align}
				From (\ref{eqn:obj_func:I_k:reserved}), (\ref{eqn:obj_func:I_k:loca}), (\ref{eqn:obj_func:I_k:ondemand_a}), and (\ref{eqn:obj_func:I_k:ondemand_b}), we have $c_{k} = \sum_l c_{k,l}$.		 
		\end{proof}

			Let $c^{\text{Alg 1}}_{k,l}$ and $c^{\text{Alg 2}}_{k,l}$ be the objective values of (\ref{eqn:new_obj_func}) for Algorithm \ref{algorithm:heu_per_plan} and Algorithm \ref{algorithm:online-deter}, respectively. Then from Lemma \ref{lemma:objective_function:seperated_cost}, we have 
			$c^{\text{Alg 1}}_{k} = \sum_l c^{\text{Alg 1}}_{k,l}$ and $c^{\text{Alg 2}}_{k} = \sum_l c^{\text{Alg 2}}_{k,l}$. In the next two lemmas, we 
			compare $c^{\text{Alg 1}}_{k,l}$ and $c^{\text{Alg 2}}_{k,l}$ for two difference cases of 
			the value of level $l$, which are then combined in Lemma \ref{lemma:expensive_phase} to provide a 
			bound on the ratio between $c^{\text{Alg 1}}_{k}$ and $c^{\text{Alg 2}}_{k}$.
			
			\begin{lemma}
				Within $I_k \in \mathcal{I}_{\text{expensive}}$, let $l_r$ be the number of VMs that Algorithm \ref{algorithm:heu_per_plan} reserves in the first time slot of $I_k$.  For $l \in \{1, \ldots, l_r \}$, $c^{\text{Alg 2}}_{k,l} \leq   {3} c^{\text{Alg 1}}_{k,l}$. \label{lemma:ex_phase:lr}
			\end{lemma}			
			\begin{proof}
				Consider any level $l \in \{1,\dots, l_r\}$.  Algorithm \ref{algorithm:heu_per_plan} reserve a VM at level $l$.  From Remark \ref{remark:local_proces}, because we only use the reserved VM at level $l$, we have,
				\begin{align}
					c^{\text{Alg 1}}_{k,l} = \gamma. \label{eqn:lemma2:c_alg1}
				\end{align}
				Furthermore, since (\ref{eqn:algorithm1:vital_ineq}) is satisfied for this level, (\ref{eqn:online_thres}) is satisfied with $t= k\tau$, so that Algorithm \ref{algorithm:online-deter} reserves a VM at level $l$ no later than $t= k\tau$. Let $t_l' \in \{(k-1)\tau + 1,\dots, k\tau\}$ denote the time the edge node reserves at level $l$ in Algorithm \ref{algorithm:online-deter}.

				\begin{itemize}
					\item[1. ] We first consider any level $l$ that is lower than the number of local VMs, i.e., $l \leq w$. Since $l \leq n_t^{\text{Alg 2}} +w, \forall t \in I_k$, from (\ref{eqn:cost:each_lvl}), there is no request allocated to on-demand VMs, i.e.,
					\begin{align}
						\sum_{t \in I_k} d_t^l p \mathbb{I} (l > n_t^{\text{Alg 2}} +w ) = 0. \label{eqn:add:a}
					\end{align}
					After a VM is reserved at $t_l'$, $n_{t}^{\text{Alg 2}} \geq l, \forall t \in \{t_l',\ldots, k\tau \}$. Hence, 
					\begin{align}
						\sum_{t = t_l'}^{k\tau} d_t^l \lambda \mathbb{I}( n_t^{\text{Alg 2}} < l \leq n_t^{\text{Alg 2}} +w ) = 0. \label{eqn:add:b}
					\end{align}
					For $1 \leq t \leq t_l' -1$, 
					\begin{align}
						\lambda \sum_{i=(k-1)\tau + 1}^{t_l'-1} d_t^l  \mathbb{I}( n_t < l \leq n_t +w ) \leq \lambda \sum_{i=(k-1)\tau + 1}^{t_l'-1} d_i^l. \label{eqn:add:d}
					\end{align}
					As a result, by substituting (\ref{eqn:add:a}), (\ref{eqn:add:b}),  and (\ref{eqn:add:d}) into (\ref{eqn:cost:each_lvl}), the cost $c^{\text{Alg 2}}_{k,l}$ are upper-bounded by 
					\begin{align}
						c^{\text{Alg 2}}_{k,l}\leq \lambda \sum_{i=(k-1)\tau + 1}^{t_l'-1} d_i^l + \gamma. \label{eqn:lemma2:1:alg2_inequality}
					\end{align}
					Moreover, since no VM is reserved at level $l$ in $I_k$ before $t_l'$, the condition (\ref{eqn:online_thres}) implies that 
					\begin{align}
						\lambda \sum_{i=(k-1)\tau + 1}^{t_l'-1} d_i^l < \gamma - (p-\lambda) \sum_{i=(k-1)\tau + 1}^{t_l'-1}d_i^{l+w}. \label{eqn:lemma2:1:alg2_inequalityb} 
					\end{align}
					As a result, from (\ref{eqn:lemma2:c_alg1}),  (\ref{eqn:lemma2:1:alg2_inequality}), and (\ref{eqn:lemma2:1:alg2_inequalityb}), we obtain the ratio
					\begin{align}
						\frac{c^{\text{Alg 2}}_{k,l}}{c^{\text{Alg 1}}_{k,l}}\leq   \frac{\gamma + \gamma - (p-\lambda) \sum_{i=(k-1)\tau + 1}^{t_l'-1}d_i^{l+w}}{\gamma} \leq 2. \label{eqn:lemma2:1}
					\end{align}
					\item[2. ] Now we consider level $l > w$.  {Recall that $t_{l-w}'$ denotes the time that the edge node reserves a VM at level  {$l-w$}. After a VM is reserved at $l-w$, $n_t^{\text{Alg 2}} \geq l-w, \forall t \in \{t_{l-w}', \ldots, k\tau \}$. The condition (\ref{eqn:online_thres}) and $d_t^l \leq d_t^{l-1}$ imply $t_{l-w}' \leq t_l'$.} Before $t_{l-w}'$, $n_t^{\text{Alg 2}} < l - w$. Therefore, for $t \in \{(k-1)\tau,\ldots, t_{l-w}' -1 \}$,
					\begin{align}
						\nonumber	\mathbb{I} (l > n_t^{\text{Alg 2}} +w ) = 1,  \\
						\mathbb{I}( n_t^{\text{Alg 2}} < l \leq n_t^{\text{Alg 2}} +w ) = 0. \label{eqn:add:e}
					\end{align}
					For $t \in  \{t_{l-w}',\ldots, t_{l}'-1 \}$, 
					\begin{align}
						\nonumber	\mathbb{I} (l > n_t^{\text{Alg 2}} +w ) = 0,  \\
						\mathbb{I}( n_t^{\text{Alg 2}} < l \leq n_t^{\text{Alg 2}} +w ) \leq 1. \label{eqn:add:f}
					\end{align}
					Here, by substituting (\ref{eqn:add:e}) and (\ref{eqn:add:f}) into (\ref{eqn:cost:each_lvl}), we conclude that $c^{\text{Alg 2}}_{k,l}$ is upper-bounded by
					\begin{align}
						c^{\text{Alg 2}}_{k,l}\leq p \sum_{i=(k-1)\tau + 1}^{t_{l-w}'-1} d_i^l +  \lambda \sum_{i=t_{l-w}'}^{t_l'-1} d_i^l + \gamma. \label{eqn:lemma2:2:alg2_inequality}
					\end{align}
					Consider the first element of the right hand-side of  (\ref{eqn:lemma2:2:alg2_inequality}). The inequality $d_i^{l} \leq d_i^{l-w}$  implies that
					\begin{align}
						p \hspace{-2mm}\sum_{i=(k-1)\tau + 1}^{t_{l-w}'-1} \hspace{-2mm} d_i^l \leq \lambda \hspace{-2mm} \sum_{i = (k-1)\tau + 1}^{t_{l-w}'-1} \hspace{-2mm} d_i^{l-w} + (p - \lambda) \hspace{-4mm}\sum_{i = (k-1)\tau + 1}^{t_{l-w}'-1} \hspace{-2mm} d_i^{l} . \label{eqn:b}
					\end{align}
					Furthermore, (\ref{eqn:online_thres}) implies that
					\begin{align}
						\lambda \hspace{-2mm} \sum_{i = (k-1)\tau + 1}^{t_{l-w}'-1} \hspace{-2mm} d_i^{l-w} + (p - \lambda) \hspace{-4mm}\sum_{i = (k-1)\tau + 1}^{t_{l-w}'-1} \hspace{-2mm} d_i^{l} < \gamma. \label{eqn:add:g}
					\end{align}
					Hence, from (\ref{eqn:b})  and (\ref{eqn:add:g}),
					\begin{align}
						p \hspace{-2mm}\sum_{i=(k-1)\tau + 1}^{t_{l-w}'-1} \hspace{-2mm} d_i^l < \gamma. \label{eqn:add:l}
					\end{align}
					Consider the second element of the right hand-side of  (\ref{eqn:lemma2:2:alg2_inequality}).
					We have
					\begin{align}
						\lambda \!\!\!\! \sum_{i=t_{l-w}'}^{t_l'-1} \!\!\! d_i^l \leq 	\lambda \!\!\!\! \sum_{i=(k-1)\tau + 1}^{t_l'-1} \!\!\! \!\!\!\! d_i^l. \label{eqn:add:h}
					\end{align}
					Moreover, since the edge node only reserves at level $l$ until $t_l'$, (\ref{eqn:online_thres})  implies
					\begin{align}
						\sum_{i=(k-1)\tau + 1}^{t_l'-1} \!\!\! \!\!\!\! d_i^l <
						\gamma - (p-\lambda) \!\!\!\!\!\!\!\sum_{i=(k-1)\tau + 1}^{t_l'-1} \!\!\!\!\!\!\! d_i^{l+w}. \label{eqn:add:i}
					\end{align}
					Since  $\sum_{i=(k-1)\tau + 1}^{t_l'-1}d_i^{l+w} \geq 0$, we have
					\begin{align}
						\gamma - (p-\lambda) \sum_{i=(k-1)\tau + 1}^{t_l'-1}  d_i^{l+w} \leq \gamma. \label{eqn:add:j}
					\end{align}		
					Hence, from 	(\ref{eqn:add:h}), (\ref{eqn:add:i}), and (\ref{eqn:add:j}), we have
					\begin{align}
						\lambda \!\!\!\! \sum_{i=t_{l-w}'}^{t_l'-1} \!\!\! d_i^l \leq \gamma. \label{eqn:add:k}
					\end{align}
					By substituting (\ref{eqn:add:k}) and (\ref{eqn:add:l}) into (\ref{eqn:lemma2:2:alg2_inequality}), we obtain the ratio
					\begin{align}
						\frac{c^{\text{Alg 2}}_{k,l}}{c^{\text{Alg 1}}_{k,l}}\leq \frac{\gamma + \gamma + \gamma}{\gamma} = 3. \label{eqn:lemma2:2}
					\end{align}
				\end{itemize}
				{As a result, from (\ref{eqn:lemma2:1}) and (\ref{eqn:lemma2:2}), $c^{\text{Alg 2}}_{k,l}\leq 3 c^{\text{Alg 1}}_{k,l}, \forall l \in \{1,\dots, l_r\}$.}
			\end{proof}	
			
			\begin{lemma}
				Within $I_k \in \mathcal{I}_{\text{expensive}}$, let $l_r$ be the number of VMs that Algorithm \ref{algorithm:heu_per_plan} reserves in the first time slot of $I_k$.  For $l \in \{ l_r+1,\ldots,  d^{\max} \} $, $c^{\text{Alg 2}}_{k,l}\leq \frac{p}{\lambda} c^{\text{Alg 1}}_{k,l}$. \label{lemma:expensive_phase_on_demand}
			\end{lemma}
			\begin{proof}
				At any levels at or above  $l_r+1$, Algorithm \ref{algorithm:heu_per_plan} does not reserve a remote VM. Thus, (\ref{eqn:algorithm1:vital_ineq}) is not satisfied. As a result, for any $t \in \{(k-1)\tau + 1,\dots, k\tau\}$, 
				(\ref{eqn:online_thres}) is also not satisfied, which implies that	Algorithm \ref{algorithm:online-deter} does not reserve any VMs. Hence, the cost resulted from Algorithm \ref{algorithm:heu_per_plan} and Algorithm \ref{algorithm:online-deter} is from remote on-demand instances and local processing. Since $\lambda < p$, the cost of Algorithm \ref{algorithm:heu_per_plan} is lower-bounded by the local processing cost, while the cost of Algorithm \ref{algorithm:online-deter} is upper-bounded by the cost of using remote on-demand instances. Thus, the ratio of $\frac{c^{\text{Alg 2}}_{k,l}}{c^{\text{Alg 1}}_{k,l}}$ is upper-bounded  by $\frac{p}{\lambda}$.
			\end{proof}

			\begin{lemma}
				Within $I_k \in \mathcal{I}_{\text{expensive}}$, $c^{\text{Alg 2}}_k \leq \max\{3 ,\frac{p}{\lambda} \} c^{\text{Alg 1}}_k$. \label{lemma:expensive_phase}
			\end{lemma}
			\begin{proof}
				From Lemma \ref{lemma:objective_function:seperated_cost}, Lemma  \ref{lemma:ex_phase:lr}, and Lemma \ref{lemma:expensive_phase_on_demand}, for $I_k \in \mathcal{I}_{\text{expensive}}$, we have 
				\begin{align}
				\frac{c^{\text{Alg 2}}_k}{c^{\text{Alg 1}}_k} \leq \max_l \left \{ \frac{c^{\text{Alg 2}}_{k,l}}{c^{\text{Alg 1}}_{k,l}} \right \}  =  \max\{3 ,\frac{p}{\lambda} \}.
				\end{align}
			\end{proof}	
		
		 From Lemma \ref{lemma:expensive_phase}, we obtain the following main proposition on the performance of Algorithm \ref{algorithm:online-deter}:
		
			\begin{proposition}
				Algorithm \ref{algorithm:online-deter} has  {$\max\{6,\frac{2p}{\lambda}\}$} competitive ratio.
			\end{proposition}
			\begin{proof}
				We have
				\begin{align}
				\nonumber	\frac{c^{\text{Alg 2}}}{c^{\text{Alg 1}}} & = \frac{\sum_{k \in \mathcal{I}_{\text{cheap}}}c^{\text{Alg 2}}_k + \sum_{k \in \mathcal{I}_{\text{expensive}}}c^{\text{Alg 2}}_k}{\sum_{k \in \mathcal{I}_{\text{cheap}}}c^{\text{Alg 1}}_k + \sum_{k \in \mathcal{I}_{\text{expensive}}}c^{\text{Alg 1}}_k}\\
				\nonumber	& \leq  \max \left \{ \frac{c^{\text{Alg 2}}_k}{c^{\text{Alg 1}}_k} | I_k \in \mathcal{I}_{\text{cheap}}, \frac{c^{\text{Alg 2}}_k}{c^{\text{Alg 1}}_k} | I_k \in \mathcal{I}_{\text{expensive}} \right \}\\
				\nonumber	& =  \max \left \{ 1, \max \{ 3 ,\frac{p}{\lambda} \}   \right   \}\\
				\nonumber	& =  \max\{3 ,\frac{p}{\lambda} \}.
				\end{align}
				From Proposition \ref{proposition:interval-optimal}, Algorithm \ref{algorithm:heu_per_plan} has 2-approximation ratio, i.e., $c^{\text{Alg 1}} \leq 2 c^*$. Hence, $c^{\text{Alg 2}} \leq  {\max\{6,\frac{2p}{\lambda}\}}c^*.$
			\end{proof}
			
			Next, we consider the performance of Algorithm \ref{algorithm:online-deter} in a special case where the local capacity at the edge node is zero. In this case, the system is reduced to the one in \cite{WangBroker2015}. However, we note that Algorithm \ref{algorithm:online-deter}  is different from the online algorithm proposed in \cite{WangBroker2015} because Algorithm \ref{algorithm:online-deter} does not consider past history in the beginning of each interval. It only considers the history within each $I_k$. This is in contrast to \cite{WangBroker2015}, where at any given $t$, the proposed online algorithm always considers historical demand from $t-\tau+1$ to $t$. Nevertheless, as shown in the following, in this case Algorithm 
			\ref{algorithm:online-deter} has the same competitive ratio as the online algorithm proposed in \cite{WangBroker2015}.
			\begin{proposition}
				When there is no edge capacity, i.e., $w=0$, Algorithm \ref{algorithm:online-deter} has $4$ competitive ratio.
			\end{proposition}

						\begin{proof}
							When $w = 0$, which means $u^{l+w} = u^{l}$, (\ref{eqn:online_thres}) becomes
							\begin{align}
								\gamma \leq p \sum_{i = (k-1)\tau + 1}^{t} d_i^{l}. \label{eqn:online_thres:no_w}
							\end{align}
							With (\ref{eqn:online_thres:no_w}), Lemma \ref{lemma:cheap_phase} still holds for $I_k \in \mathcal{I}_{\text{cheap}}$. 
							
							Let us reconsider Lemma \ref{lemma:ex_phase:lr}, for $I_k \in \mathcal{I}_{\text{expensive}}$ and $l \leq l_r$. Since $w=0$, $l>w$. As a result, (\ref{eqn:lemma2:2:alg2_inequality}) becomes
							\begin{align}
								c^{\text{Alg 2}}_{k,l}\leq p \sum_{i=(k-1)\tau + 1}^{t_l'} d_i^l + \gamma. \label{eqn:add:m}
							\end{align}
							Since there is no reservation at level $l$ until $t_l'$,
							\begin{align}
								p \sum_{i = (k-1)\tau + 1}^{t_l'} d_i^{l} < \gamma . \label{eqn:add:n}
							\end{align}
							Moreover, for $1 \leq l \leq l_r$, Algorithm \ref{algorithm:heu_per_plan} only use reserved VM, i.e., $c^{\text{Alg 1}}_{k,l} = \gamma$. Hence, by substituting (\ref{eqn:add:n}) into (\ref{eqn:add:m}), we have
							\begin{align}
								\frac{c^{\text{Alg 2}}_{k,l}}{c^{\text{Alg 1}}_{k,l}}\leq 2.
							\end{align}

							Let us reconsider Lemma \ref{lemma:expensive_phase_on_demand}, for $I_k \in \mathcal{I}_{\text{expensive}}$ and $l > l_r$. Since both algorithms reserve no VM from $l=l_r +1$	 to $d^{\max}$ and there is no edge processing, the objective values of (\ref{eqn:new_obj_func}) of both algorithms are only resulted from on-demand VMs. Hence, $c^{\text{Alg 2}}_{k,l}= c^{\text{Alg 1}}_{k,l}$. 
							
							Combining the above, we have
							\begin{align}
								c^{\text{Alg 2}} \leq 2 c^{\text{Alg 1}} \leq 4c^*.
							\end{align}	
						\end{proof}
							\begin{table*} [t]
								\caption{Power consumption of  physical processors corresponding to EC2 Instance offers.\protect\footnotemark[1]}
								\begin{center}
									\footnotesize
									\begin{tabular}{ |c|c|c|c| } 
										\hline
										\bf	Instance type & \bf Equivalent Physical Processor & \bf Full Load   & \bf Hourly On-demand Price \\
										\hline
										{m3.medium} & {Intel Xeon E5-2670} & {305W} &  \$0.067 \\ 
										\hline
										{c3.large} & {Intel Xeon E5-2680} & {425W} &  \$0.105 \\ 
										\hline			
									\end{tabular} \label{tab:Amazon_hardware}	
								\end{center}	
							\end{table*}							
		\begin{figure*}
			\centering
			\hspace{-3cm}\subfigure[Group: 1]{\includegraphics[width=.25\linewidth]{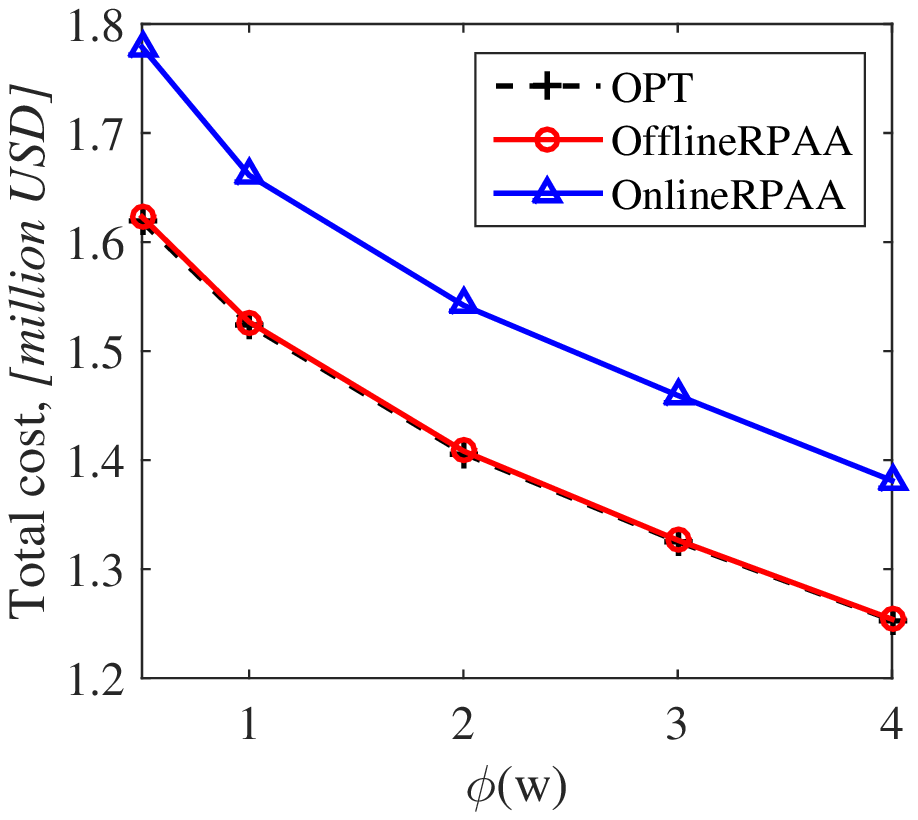}}
			\subfigure[Group: 2]{\includegraphics[width=.25\linewidth]{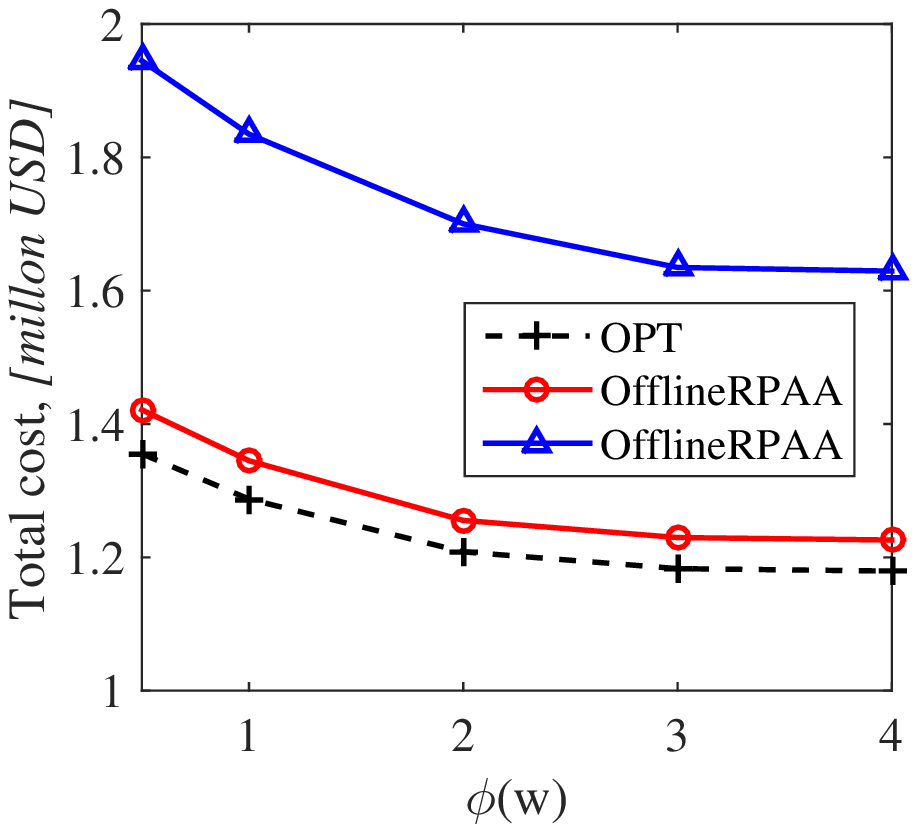} } 
			\subfigure[Group: 3]{\includegraphics[width=.25\linewidth]{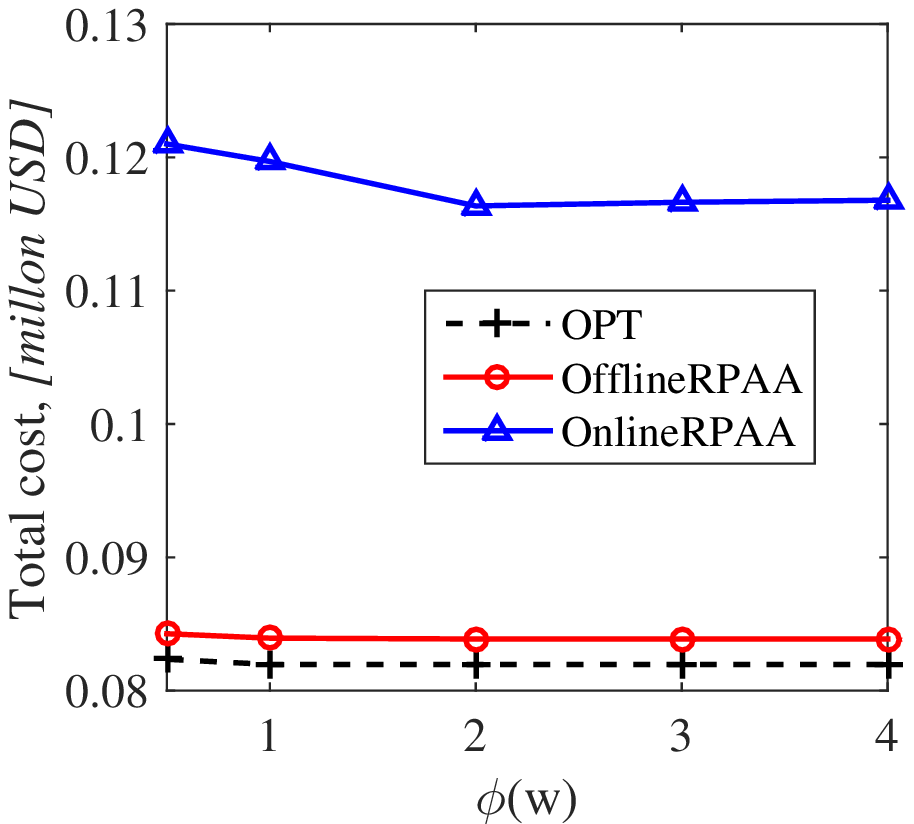}}\hspace{-2cm}
			\caption{The impact of the edge node's capacities  {on the total cost}. \label{fig:alg_per_w_change}}
		\end{figure*}			

			\section{Numerical Results} \label{section:sim}			
			Besides the approximation and competitive ratios derived in the previous sections, we numerically evaluate the performance of the proposed algorithms with extensive simulation based on the parameters specified in Amazon pricing policies \cite{Amazon:price} and Google cluster-usage traces \cite{clusterdata:Reiss2011}.  We set $p' = \$ 0.067$ and $\gamma = \$ 1.0452$.  We also consider   Amazon's reserved instance ``m3.medium'' whose equivalent physical processor is Intel Xeon E5-2670.\footnote[1]{ {The equivalent processors  {are according to} \cite{AWSforMD}.} The power consumption of processors is measured under stress tests in \cite{Chi2014,Jarp:1434748}.}  From  \cite{Chi2014,Jarp:1434748}, its power consumption in full utilization is $305$ W, as shown in Table \ref{tab:Amazon_hardware}.  The edge processing cost $\lambda'$ is set at $\$ 0.03$ based on electricity usage,  {assuming the electricity price at $\$0.1$ per $1$ kWh}.\footnote[2]{The electricity price in US in Oct. 2017 is from \cent $ 8.20$ to \cent $15.40$ \cite{US:electricity}.}  Unless otherwise specified, the default value of the edge's capacity is set at one standard deviation of the demand, and the default effective reservation time $\tau$ is one week. Google cluster-usage traces were measured in one month, i.e., $T= $ one month ($4$ weeks). With these pricing structures, even though the edge processing is computed only from electricity cost, the competitive ratio is $6$ with ``m3.medium'' and ``c3.large''. As explained in Section \ref{section:sys_model}, the value of $\theta$ is not important. Therefore, without loss of generality, we set $\theta = 0$.


		\subsection{Comparison Targets}

		We compare the performance of the two proposed algorithms with the following baselines. 
		\begin{itemize} 
			\item[1.] Edge and On-Demand: Here, we process VM requests first at the edge (see Remark \ref{remark:local_proces}). The excess requests are offloaded to the cloud with  remote on-demand VMs. This strategy is labeled as ``E+Od''.
			\item[2.] Wang \emph{et. al}: Here, we apply the online algorithm proposed by Wang  \emph{et. al} in \cite{WangBroker2015} where the authors only consider remote reserved VMs and remote on-demand VMs. This algorithm is labeled as ``Wang''.
			\item[3.] Edge + Wang \emph{et. al}: Here, we process VM requests first at the edge. The excess requests are allocated by following the algorithm ``Wang'' above. This strategy is labeled as ``E+Wang''.
			\item[4.] On-Demand Only: Here, all VM requests are served by remote on-demand VMs.
		\end{itemize}
		 {These algorithms are compared in different scenarios where demands have different fluctuation levels, in order to reveal} the effects of the edge node's capacities and the cloud node's pricing structure on the performance of these algorithms. We also   investigate the impact of the reservation period on the algorithms' performance.
			\subsection{Google Cluster-Usage Traces}
			{Since the workload information in public clouds is often confidential}, we  use Google cluster-usage traces \cite{clusterdata:Reiss2011} to examine the proposed algorithms in practical scenarios.  {We assume that Google's computing demands approximate public IaaS servers' demands \cite{WangBroker2015}}. Google recorded tasks arriving at  {one of} its  server clusters  {of about 12500 physical machines} within one month in May 2011.  {Here, we use the revised data, version $2.1$, which was updated  {on} Nov. 11, 2014. Since user names are encrypted by strings of characters, we use function ``as.factor'' in  {the R programming language} to determine the number of different strings. We find $901$ users within the trace period. After that, we use function ``as.numeric'' in  {R} to produce a one-to-one mapping from strings to numbers.  {As a result, user names are converted from strings to numbers. Tasks arrive in the time scale of $\mu s$, while the billing cycle of on-demand VMs is one hour. For simplicity, each task is assumed to be equivalent to one VM request, which takes one hour to be processed. We also assume that an instance is required to serve each VM request. Therefore, for each user, its demand curve is computed by counting the number of tasks arriving in each hour.} We then analyze the users' demand within the month. 	Similar to \cite{WangBroker2015}, we also divide users into $3$ groups based on their demand fluctuation level, i.e., the ratio between the user demand's standard deviation and mean.
					\begin{figure*}
						\centering
						\subfigure[Group: 1\hspace{1.7cm}]{\includegraphics[width=.37\linewidth]{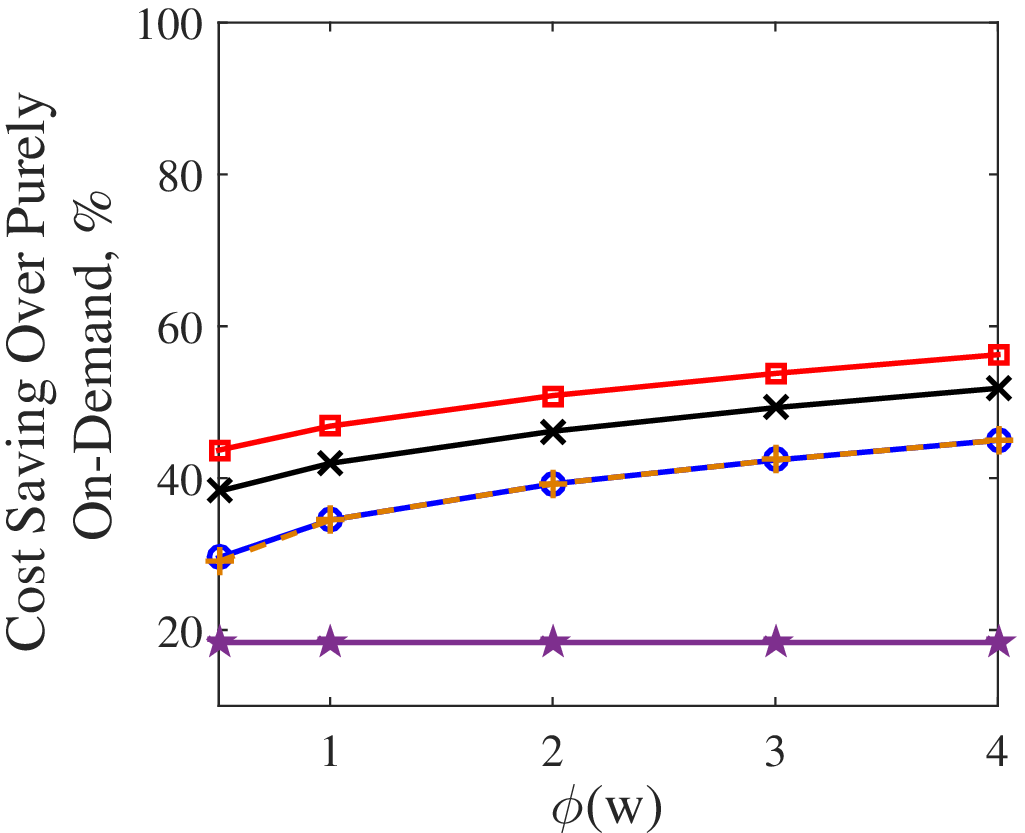}}\hspace{-2.1cm}
						\subfigure[Group: 2\hspace{1.7cm}]{\includegraphics[width=.37\linewidth]{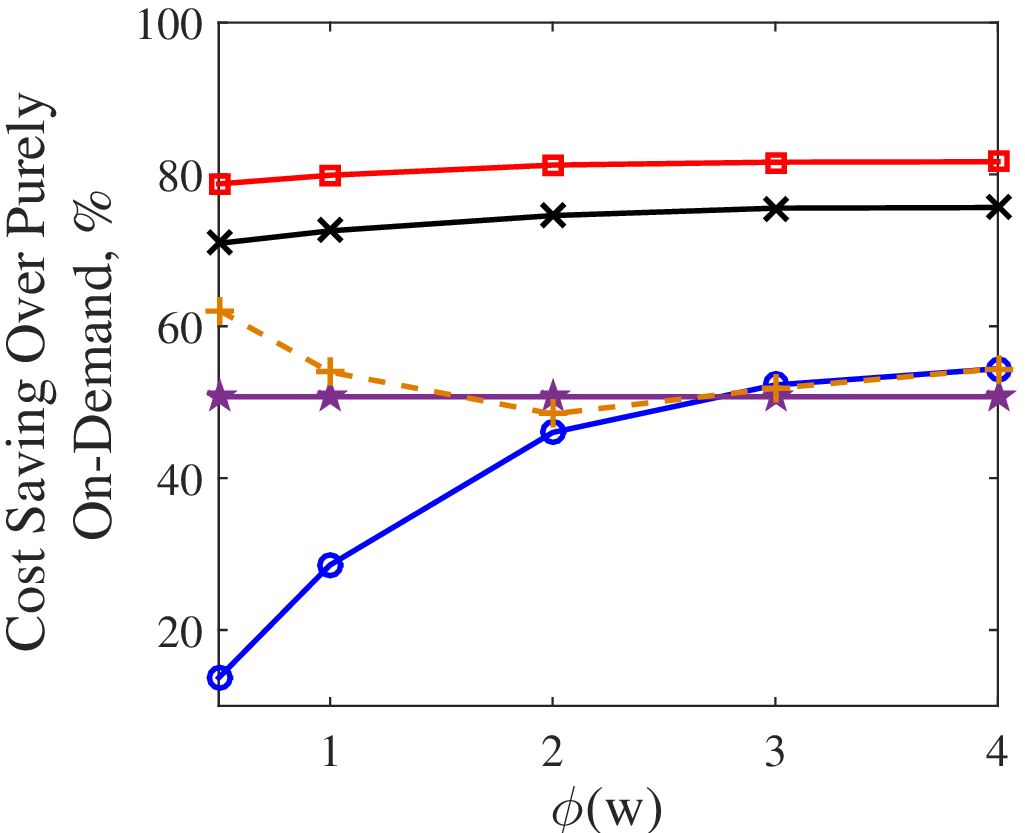}\label{fig:GG_demand_dis_w_change:90_2} } \hspace{-2.1cm}
						\subfigure[Group: 3\hspace{1.7cm}]{\includegraphics[width=.37\linewidth]{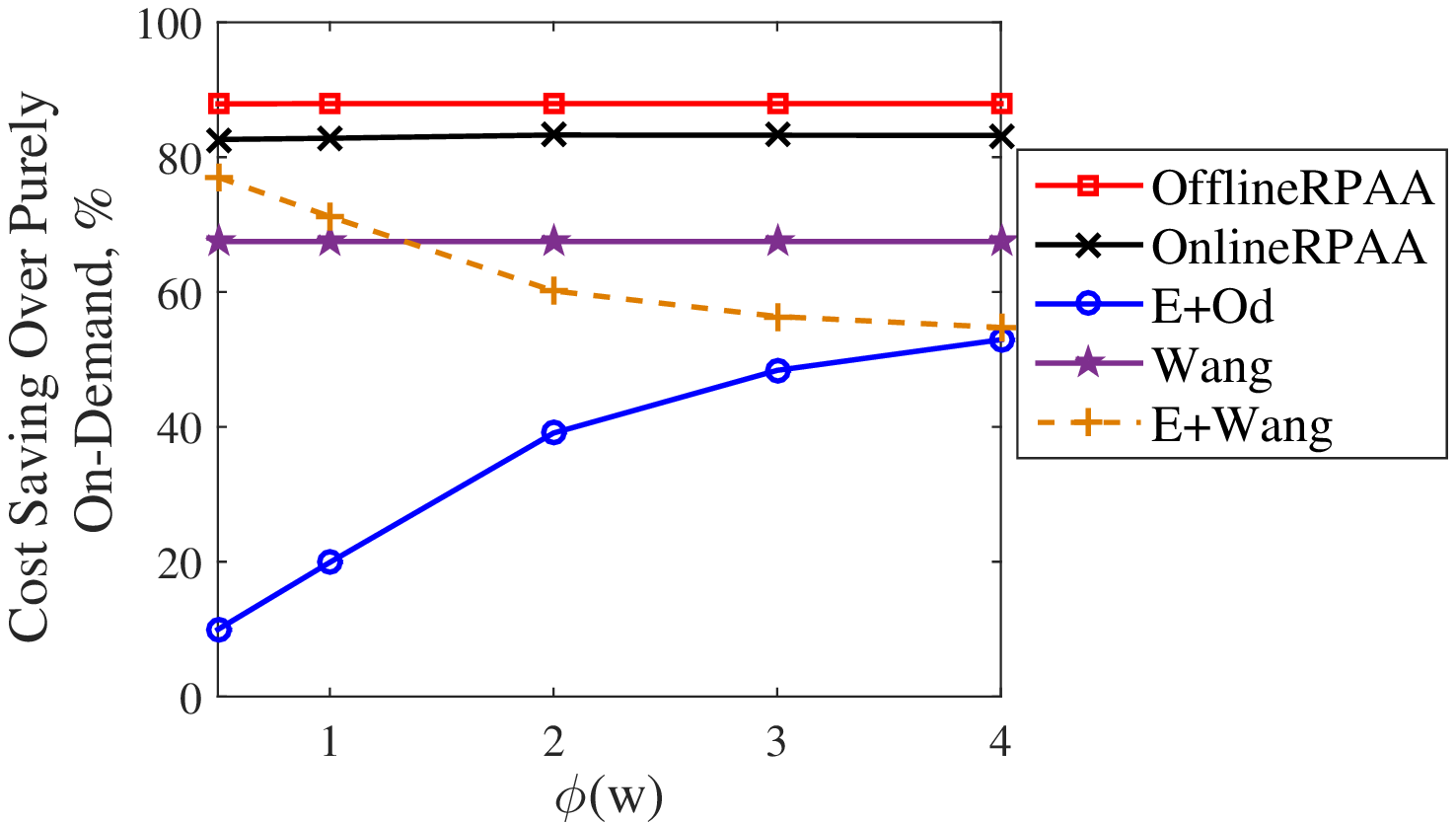}}
						\caption{The impact of the edge node's capacities  {on the cost saving percentage of algorithms over purely assigning requests to on-demand instances}. \label{fig:GG_demand_dis_w_change}}
					\end{figure*}
					\begin{figure}
						\centering
						\includegraphics[width=.8\linewidth]{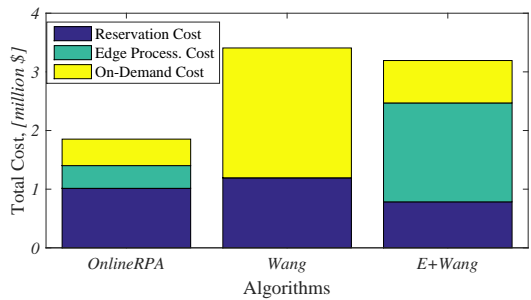}
						\caption{The cost distribution of the three algorithms with $\tau =$ 1 week at Group 2. \label{fig:Cost_allocation}}
					\end{figure}			
				\begin{itemize}
					\item[1. ] Group 1 (High Fluctuation): Users in this group have demand fluctuation level greater than $5$. There are $570$ users in this group.
					\item[2. ] Group 2 (Medium Fluctuation): Users in this group have demand fluctuation levels between 1 and 5. There are $308$ users in this group.
					\item[3. ] Group 3 (Low Fluctuation): Users in this group have demand fluctuation levels less than 1. There are $23$ users in this group.
				\end{itemize}
				
				{The edge node simply adds up all users' demand as the aggregate demand. Since the data is recorded in one month, the length of the aggregate demand vector is  {$672$ timeslots in length}, which is equivalent to the number of hours in $4$ weeks.}

			\subsection{Impact of Edge Node's Capacity}			

						
				In this section, we investigate the impact of the edge node's capacity. Here, the algorithms' performance is evaluated based on their cost savings over purely allocating VM requests to remote on-demand VMs. For each group, we consider the aggregate demand of the group. Let $\sigma$ denote the standard deviation of the aggregate demand, and $\phi(w) = \frac{w}{\sigma}$ be the ratio of edge node's capacity over the demand standard deviation. We consider different edge node  capacities, for five different values of $\phi \in [0.5~1~2~3~4]$.  {We first compare the performance of the two proposed algorithms with the optimal solution found by the Branch-and-Bound algorithm. Then, we study the cost saving of the proposed algorithms over assigning on-demand VMs, in comparison with that of the baselines.}
		
		 {Fig. \ref{fig:alg_per_w_change} shows that the proposed offline algorithm has near optimal performance, while the online algorithm also performs well. Moreover, we observe the importance of edge computing capacity. Fig. \ref{fig:alg_per_w_change} suggests that the faster the demands fluctuate, the faster the total cost decays. Thus, the local VMs at the edge node serves to smooth out the fluctuation of the demands, since their usage cost is in between that of on-demand VMs and reserved VMs..}
				
			We observe in Fig. \ref{fig:GG_demand_dis_w_change} that, firstly, ``OfflineRPAA'' and ``OnlineRPAA'' outperform all alternatives.  Moreover, the comparison between  ``OnlineRPAA'' and ``E+Wang''  suggests that an effective resource allocation algorithm should consider the edge node's parameters instead of a trivial solution such as ``E+Wang''. Finally, the cost savings of ``OfflineRPAA'' and ``OnlineRPAA''  over the pure on-demand strategy increase as the edge's capacity increases, especial for users in Group 1, which demonstrates the benefit of effective utilization of edge computing.

			Fig. \ref{fig:Cost_allocation} explains why ``OnlineRPAA'' performs significantly better  than ``Wang'' and ``E+Wang'', using as an example users in Group $2$. As in Fig. \ref{fig:Cost_allocation}, ``Wang'' not only reserves more than ``OnlineRPAA'' but also has much higher on-demand cost. It implies that many of   reserved VMs in ``Wang'' are under utilized, while ``OnlineRPAA''   takes advantage of edge  VMs to reduce its cost. On the other hand, in ``E+Wang'', VM requests are allocated to the edge  VMs first and then the excess ones followed ``Wang''. It is observed that the cost of reservation of ``E+Wang'' is slightly less than that of ``OnlineRPAA''. However, the  on-demand cost and edge processing cost of ``E+Wang'' are higher than those of ``OnlineRPAA''. Thus, we conclude that ``OnlineRPAA'' utilizes resources more efficiently than  ``E+Wang''.
			
			\subsection{Impact of Reservation Period}
			\begin{figure*}[t]
				\centering
				\subfigure[Group: 1\hspace{1.7cm}]{\includegraphics[width=.37\linewidth]{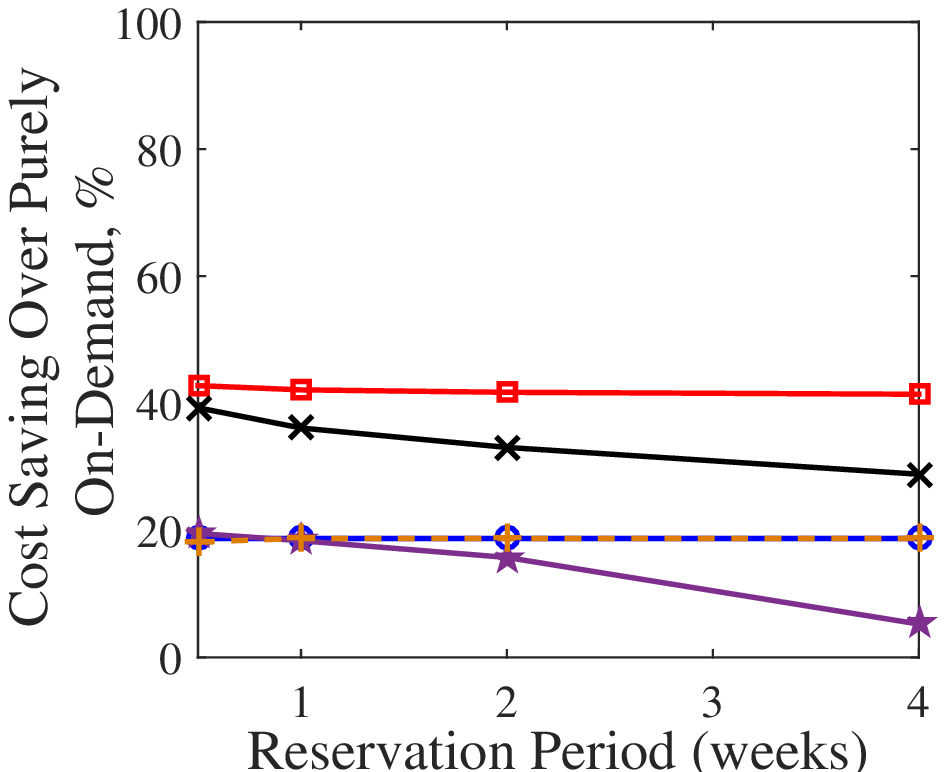}}\hspace{-2cm}
				\subfigure[Group: 2\hspace{1.7cm}]{\includegraphics[width=.37\linewidth]{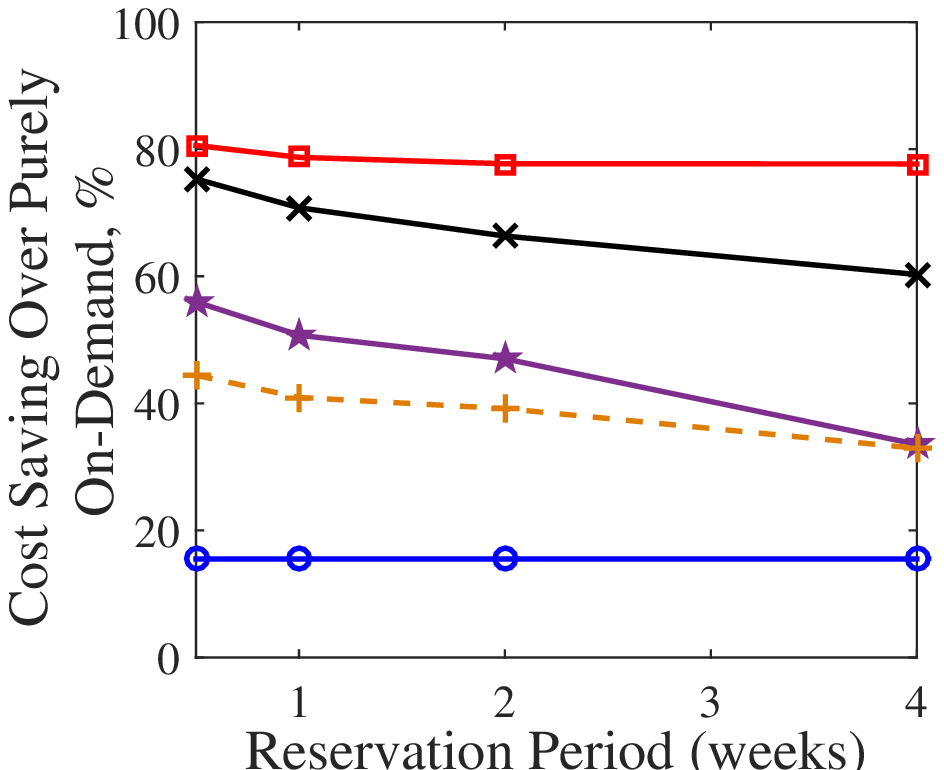}}\hspace{-2cm}
				\subfigure[Group: 3\hspace{1.7cm}]{\includegraphics[width=.37\linewidth]{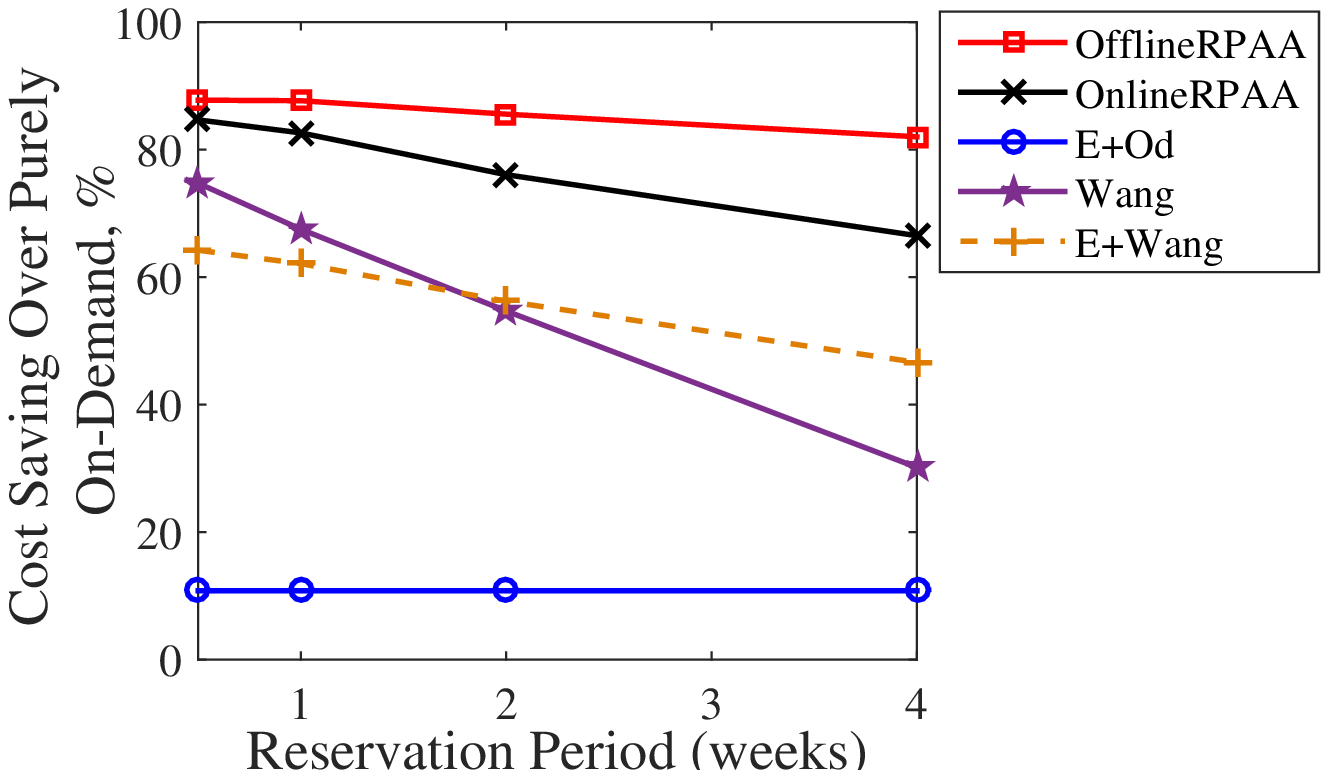} \label{fig:GG_demand_dis_tau_change:group3:90}}
				\caption{The impact of reservation periods  {on the cost saving percentage of algorithms over purely assigning requests to on-demand instances}. \label{fig:GG_demand_dis_tau_change}}
			\end{figure*}			
						\begin{figure}
							\centering
							\includegraphics[width=.8\linewidth]{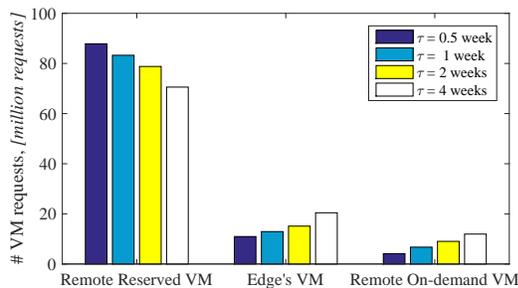} 
							\caption{VM requests allocation of OnlineRPAA with different reservation periods at Group 2. \label{fig:alg_2_taskall}}
						\end{figure}
			In this section, we  investigate the impact of the reservation period on algorithm performance. We further assume that the reservation cost $\gamma$ increases proportionally with respect to the reservation period. We observe from Fig. \ref{fig:GG_demand_dis_tau_change} that the  algorithms' performance decrease when the reservation period increases. This phenomenon is explained in Fig. \ref{fig:alg_2_taskall}, where we investigate the resource allocation of ``OnlineRPAA'' in different reservation periods, with users in Group $2$ as an example. We observe that, as the reservation period increases, the number of VM requests assigned to remote reserved VMs decreases, while the number of requests allocated to the other two options increases. This implies that the number of reservations is reduced. This is reasonable since the upfront cost $\gamma$ increases as $\tau$ increases. Since the edge node reserves less, the performance gain of the proposed algorithms  over the pure on-demand strategy decreases. Furthermore, as shown in Fig. \ref{fig:GG_demand_dis_tau_change:group3:90}, the degradation of ``Wang'' is faster than ``OnlineRPAA' and ``E+Wang''. This again confirms the importance of considering the edge's capacity.

			\section{Conclusion} \label{section:conclusion}		
			
			In this work, a hybrid edge-cloud system is investigated. We consider an edge node that has finite computing capacity and a cloud node that offers remote computing instances under both on-demand and reservation options.  The edge node decides how to acquire computing resource from the cloud node and allocate its local and acquired resources to reduce its cost of serving users demands. An offline resource procurement and allocation solution is proposed with the prior knowledge of future demand. We then  propose an online resource procurement and allocation algorithm, which makes irrevocable decision without knowledge of future demand. For both algorithms, the worst-case performance with respect to the offline optimum is provided.  Numerical results show the importance of considering the edge node's computing capacity. Firstly, the existence of computing capacity at the edge can significantly  reduce its cost. Secondly, we observe that under typical cloud and edge pricing structure, the proposed online algorithm, which considers the edge's cost and capacity, significantly outperforms alternative solutions, including one that always processes user requests first at the edge. Finally, when the reservation period and the upfront cost are increased proportionally,  the performance of the algorithms decreases. However, the degradation of the algorithms'  {performance lessens} if the local capacity of the edge node is increased.

			\bibliographystyle{IEEEtran}
			\bibliography{IEEEabrv,references}

		\end{document}